\documentclass[12pt]{article}
\usepackage{graphicx}%
\usepackage{multirow}%
\usepackage{amsmath,amssymb,amsfonts}%
\usepackage{amsthm}%
\usepackage{mathrsfs}%
\usepackage[title]{appendix}%
\usepackage{xcolor}%
\usepackage{textcomp}%
\usepackage{manyfoot,natbib}%
\usepackage{booktabs}%
\usepackage{algorithm}%
\usepackage{algorithmicx}%
\usepackage{algpseudocode}%

\usepackage{listings}%

\usepackage{lmodern,anyfontsize}
\usepackage{bm,bbm}
\usepackage{fullpage}
\usepackage{tablefootnote}
\usepackage{booktabs}
\usepackage{threeparttable}
\usepackage{tabularx}

\def\T{\intercal}
\def\E{\mathbb {E}}
\def\Pr{\mathbb {P}}
\def\var{\mathrm {var}}

\def\cov{\mathrm {cov}}

\def\diag{\mathrm {diag}}

\def\vec{\mathrm {Vec}}

\def\sgn{\mathrm {sgn}}
\def\tr{\mathrm {tr}}
\def\mR{\mathbb{R}}

\def\card{\mathrm {card}}

\def\A{{\bm A}}
\def\B{{\bm B}}

\def\D{{\bm D}}

\def\I{{\bm I}}
\def\U{{\bm U}}
\def\M{{\bm M}}
\def\bfS{{\bm S}}

\def\V{{\bm V}}
\def\W{{\bm W}}
\def\X{{\bm X}}

\def\Z{{\bm Z}}

\def\f{{\bm f}}

\def\x{{\bm x}}
\def\y{{\bm y}}
\def\z{{\bm z}}
\def\v{{\bm v}}
\def\u{{\bm u}}
\def\e{{\bm e}}

\def\bmu{{\bm \mu}}

\def\bms{{\bm \Sigma}}
\def\G{{\bm \Gamma}}
\def\L{{\bm \Lambda}}

\def\bmzero{{\bm 0}}

\def\bmu{{\bm \mu}}

\def\bmD{{\bm \Delta}}

\def\bms{{\bm \Sigma}}
\def\G{{\bm \Gamma}}
\def\L{{\bm \Lambda}}

\def\bmzero{{\bm 0}}

\newtheorem{theorem}{Theorem}
\newtheorem{proposition}[theorem]{Proposition}%
\newtheorem{assumption}{Assumption}%
\newtheorem{lemma}{Lemma}

\newtheorem{remark}{Remark}%

\begin{document}

\title{High dimensional  matrix estimation through elliptical factor models}

 \author{Xinyue Xu$^1$, Huifang Ma$^1$, Hongfei Wang$^2$ and Long Feng$^1$ \\
School of Statistics and Data Science,\\
 LEBPS, KLMDASR and LPMC, Nankai University$^1$\\
 School of Statistics and Data Science, Nanjing Audit University$^2$}


\maketitle

\begin{abstract}
Elliptical factor models play a central role in modern high-dimensional data analysis, particularly due to their ability to capture heavy-tailed and heterogeneous dependence structures. Within this framework, Tyler’s M-estimator \citep{Tyler1987} enjoys several optimality properties and robustness advantages. In this paper, we develop high-dimensional scatter matrix, covariance matrix and precision matrix estimators grounded in Tyler’s M-estimation. We first adapt the Principal Orthogonal complEment Thresholding (POET) framework \citep{fan2013large} by incorporating the spatial-sign covariance matrix as an effective initial estimator. Building on this idea, we further propose a direct extension of POET tailored for Tyler’s M-estimation, referred to as the POET-TME method. We establish the consistency rates for the resulting estimators under elliptical factor models. Comprehensive simulation studies and a real data application illustrate the superior performance of POET-TME, especially in the presence of heavy-tailed distributions, demonstrating the practical value of our methodological contributions.

{\it keywords}: Covariance matrix estimation; High dimensional data; Spatial-sign; POET.
\end{abstract}

\section{Introduction}\label{sec:introduction}
Estimating the covariance matrix is a fundamental problem in multivariate statistical analysis. Accurate estimation plays a crucial role in numerous tasks such as principal component analysis (PCA), clustering, and discriminant analysis. When the data dimension $d$ is fixed and the sample size $n$ is large, the sample covariance matrix performs efficiently under sub-Gaussian assumptions. However, as the dimensionality grows and $d$ becomes comparable to or even larger than $n$, the sample covariance matrix becomes unstable and non-invertible. To address this issue, a substantial body of research has proposed various regularized or sparse covariance estimators, leveraging structural assumptions such as sparsity in the true covariance matrix \citep{BickelLevina2008a,BickelLevina2008b,CaiLiu2011,CaiMaWu2015,RLZ2009threshold}. See an overview in \cite{fan2016overview}.

Most of the aforementioned methods for covariance matrix estimation rely on the assumption that the covariance matrix is sparse. While this assumption is plausible in many practical settings, it does not always hold. For instance, financial assets often share common market factors, leading to strong correlations among their returns. Similarly, genes within the same biological pathway may be co-regulated by a limited number of transcriptional factors, resulting in highly correlated gene expression patterns. Moreover, when genes are stimulated by cytokines, their expressions tend to exhibit strong co-movements. In such scenarios, the sparsity assumption becomes clearly unrealistic. To address the limitations of the sparsity assumption, \citet{fan2013large} proposed an approximate factor model and developed the Principal Orthogonal complEment Thresholding (POET) framework. The POET estimator decomposes the covariance matrix into a low-rank component, capturing the pervasive dependence driven by a few common factors, and a sparse idiosyncratic component that represents weak local correlations. This unified framework encompasses several popular estimators as special cases, including the sample covariance matrix, the factor-based covariance matrix, the thresholding estimator, and the adaptive thresholding estimator. Despite its theoretical and computational appeal, the original POET procedure is constructed based on the sample covariance matrix, which can be highly sensitive to heavy-tailed distributions and outliers. Consequently, subsequent research has focused on developing robust extensions of the POET framework, such as those based on rank-based or M-estimation approaches, to enhance performance under non-Gaussian or contaminated data settings (see, e.g., \citealp{FLW2018generalPOET}).

Elliptical distributions have become a central modeling tool in modern multivariate analysis due to their ability to capture heavy-tailed behavior, tail dependence, and deviations from Gaussianity. They are widely used in financial return modeling \citep{zheng2011estimation,HanLiu2014}, gene expression data analysis \citep{avella-medina2018robust,ZhaoChenWang2024}, and neuroimaging studies \citep{HL2018ECA}. When the dimension is fixed, numerous robust estimators of the scatter matrix are available for elliptical distributions, including Tyler's M-estimator \cite{Tyler1987}, the spatial-sign covariance matrix, and multivariate Kendall's tau estimators \citep{oja2010multivariate}, all of which enjoy desirable properties such as affine equivariance and robustness to heavy tails.

However, in high-dimensional regimes where the dimension exceeds the sample size, classical robust scatter estimators become singular, non-invertible, or computationally unstable, which necessitates new approaches. Considerable effort has been dedicated to developing regularized scatter and covariance matrix estimators under elliptical models. For sparse covariance estimation, \cite{sun2014regularized} introduced a regularized Tyler's M-estimator, and \cite{goes2020robust} further enhanced this approach via thresholding techniques.

Parallel to Tyler-based methods, robust high-dimensional scatter estimation has also progressed through the study of spatial-sign covariance matrices, which naturally encode directional information and possess strong robustness properties. Recent contributions include linear shrinkage estimators \citep{raninen2021linear,raninen2021bias,ollila2022regularized,ollila2024linear} specifically designed for high-dimensional settings, as well as sparse precision matrix estimators constructed from spatial-sign covariance matrices \citep{lu2025robust}. These works demonstrate that combining sign-based methods with regularization can yield stable and interpretable estimators even under heavy-tailed distributions.

For nonsparse high-dimensional elliptical structures, factor-based modeling has emerged as a powerful alternative. \cite{FLW2018generalPOET} generalized the POET framework to elliptical factor models by leveraging the multivariate Kendall's tau matrix, providing a robust means of extracting latent low-rank dependence. More recently, \cite{ding2025sub} proposed an idiosyncratic-projected self-normalization method based on the spatial-sign covariance matrix, which offers improved performance under heavy-tailed noise and complex dependence structures.

In this paper, we extend the POET framework to the setting of spatial-sign covariance 
matrices for nonsparse high-dimensional factor models, named as POET-SS. We first establish the consistency rates of the proposed estimators under mild regularity conditions. Although the spatial-sign covariance matrix is robust to heavy-tailed distributions, its nonsparse structure leads to a  approximation bias with respect to the true covariance matrix, which may affect the efficiency of the resulting estimators. To remedy this issue, we further incorporate Tyler's M-estimator by using the spatial-sign covariance estimator as an initial input, and propose a new 
two-step procedure, termed POET-TME. Simulation studies and real data 
analysis show that the proposed POET-TME method achieves the best overall 
performance for high-dimensional elliptical factor models.

This paper is organized as follows. Section 2 introduces the POET-SS method and presents the corresponding theoretical results. Section 3 develops the POET-TME procedure. Section 4 reports the simulation studies, and Section 5 provides a real data application. All technical proofs are collected in the Appendix.

\textit{Notation:} We use the following notation throughout the paper. For any matrix $\M=[M_{jk}]\in\mR^{d\times d}$, define the matrix entry-wise maximum value , Frobenius norm, spectral norm, elementwise $\ell_1$ norm, $\ell_1$ norm, $\ell_{\infty}$ norm and relative Frobenius norm as $\Vert \M\Vert_{\max}=\max|M_{jk}|$, $\Vert\M\Vert_F=(\sum_{jk}M_{jk}^2)^{1/2}$,  $\Vert\M\Vert_2=\lambda_{\max}^{1/2}(\M^{\T}\M)$, $\Vert\M\Vert_{1,1}=\sum_{jk}|M_{jk}|$, $\Vert\M\Vert_1=\max_k\sum_j|M_{jk}|$, $\Vert\M\Vert_\infty=\max_j\sum_k|M_{jk}|$ and $\Vert\M\Vert_{\bms}=d^{-1/2}\Vert\bms^{-1/2}\M\bms^{-1/2}\Vert_F$. Let $\lambda_j(\M)$ be the $j$-th largest eigenvalue of $\M$. If there are ties, $\lambda_j(\M)$ is any one of
the eigenvalues such that any eigenvalue larger than it has rank smaller than $j$, and any
eigenvalue smaller than it has rank larger than $j$. Let $u_j(\M)$ be any unit vector $\v$ such that $\v^\intercal\M\v=\lambda_j(\M)$. Without loss of generality, we assume that the first nonzero entry of $u_j(\M)$ is positive. For any two index sets $\mathcal{S}_1$ and $\mathcal{S}_2$, use $\M_{\mathcal{S}_1\mathcal{S}_2}$ to denote the $|\mathcal{S}_1|\times|\mathcal{S}_2|$ matrix with rows and columns of $\M$ indexed by $\mathcal{S}_1$ and $\mathcal{S}_2$, respectively. 
The notation $\M\succ 0$ indicates that $\M$ is positive definite.
For any vector $\v=(v_1,\ldots,v_d)^\intercal\in\mR^d$, define the $\ell_0$, $\ell_q(0<q<\infty)$ and $\ell_\infty$ vector norms as $\Vert\v\Vert_0=\card(\{j:v_j\ne 0\})$, $\Vert \v\Vert_q=\left(\sum_i|v_i|^q\right)^{1/q}$ and $\Vert \v\Vert_\infty=\max_i|v_i|$. For any random variable $X\in\mR$, define the sub-Gaussian $(\Vert\cdot\Vert_{\psi_2})$ and sub-exponential norms $(\Vert\cdot\Vert_{\psi_1})$ of $X$ as $\Vert X\Vert_{\psi_2}=\sup_{k\ge 1}k^{-1/2}\{\E(|X|^k)\}^{1/k}$ and $\Vert X\Vert_{\psi_1}=\sup_{k\ge 1}k^{-1}\{\E(|X|^k)\}^{1/k}$. 
Let $\mathbb{S}^{d-1}=\{\v\in\mR^d:\Vert\v\Vert_2=1\}$ be the $d$-dimensional unit sphere, $\I_d$ represent the $d$ by $d$ identity matrix, and $\mathbbm{1}(\cdot)$ be the indicator function. 

\section{Large Matrix Estimation based on Spatial-sign}
\subsection{Generic POET for elliptical distribution}
Suppose $\X_1,\ldots,\X_n$ are i.i.d. from $\X$, the elliptical distribution $EC_d(\bmu,\bms,\xi)$, i.e.
\begin{align}\label{elliptical model}
    \X\mathop{=}\limits^{d}\bmu+\xi\A\U,
\end{align}
where $\bmu\in\mR^d$, $\U$ is a uniform random vector on the unit sphere in $\mR^q$, $\xi\ge 0$ is a scalar random variable independent of $\U$, and $\A\in\mR^{d\times q}$ is a deterministic matrix satisfying $\A\A^\T=\bms$. Here, $\bms$ is called the scatter matrix. Note that the representation in \eqref{elliptical model} is not identifiable since we can rescale $\xi$ and $\A$. To make the model identifiable, we require $E(\xi^2)=q<\infty$ so that $\cov(\X)=\bms$. In addition, we assume $\bms$ is nonsingular, that is, $q=d$. In this paper, we only consider continuous elliptical distributions with $\Pr(\xi=0)=0$. Here, we define $\bms_0=d\bms/\tr(\bms)$. Some literature directly assumes that $\tr(\bms)=d$ for simplicity, as seen in \citet{HL2018ECA}. Hence, in the subsequent sections, we also refer to the matrix $\bms_0$ as the scatter matrix.  

The scatter matrix $\bms_0$ plays a central role in various applications. For example, the eigenstructure of $\bms_0$ governs the principal component analysis (PCA) of $\X$. In financial applications, the inverse of the scatter matrix, $\bms_0^{-1}$, is used in constructing the global minimum variance portfolio (MVP). Financial returns usually exhibit a factor structure with pervasive factors. Therefore, we consider the following spiked covariance assumption.
\begin{assumption}\label{assumption 1}
    Let $m$ be a fixed constant that does not change with $n$ and $d$ such that $m \le n \le d$. As $d\to\infty$, $\lambda_1(\bms)>\lambda_2(\bms)>\cdots>\lambda_m(\bms)\gg \lambda_{m+1}(\bms)\ge\cdots\ge\lambda_d(\bms)>0$, where the spiked eigenvalues are linearly proportional to dimension $d$ while the nonspiked eigenvalues are bounded, that is, $c_0 \le \lambda_j(\bms) \le C_0$, $j>m$ for constants $c_0,C_0 > 0$. 
\end{assumption}
Assumption \ref{assumption 1} describes statistical factor models with spiked eigenvalues; see, e.g., \cite{fan2013large}, \cite{FLW2018generalPOET},  \cite{ding2021}, \cite{ding2025sub}.
Assumption \ref{assumption 1} divides the eigenvalues into the diverging and bounded ones. For simplicity, we only consider distinguishable eigenvalues (multiplicity 1) for the largest $m$ eigenvalues. Under Assumption \ref{assumption 1}, we consider the spectral decomposition of $\bms_0$:
\begin{align}\label{spectral decomposition}
    &\bms_0=\sum_{j=1}^m\lambda_j(\bms_0)u_j(\bms_0)u_j(\bms_0)^\T+\sum_{j=m+1}^d\lambda_j(\bms_0)u_j(\bms_0)u_j(\bms_0)^\T=\G_m\L_m\G_m^\T+\bms_{0u},
\end{align}
where $$\G_m=(u_1(\bms_0),\ldots,u_m(\bms_0)), \L_m=\diag(\lambda_1(\bms_0),\ldots,\lambda_m(\bms_0))$$ and $\bms_{0u}=\sum_{j=m+1}^d\lambda_j(\bms_0)u_j(\bms_0)u_j(\bms_0)^\T$. Inspired by \citet{FLW2018generalPOET}, given three initial pilot estimators $\hat{\bms}_0$, $\hat{\L}_m$, $\hat{\G}_m$ for true scatter matrix $\bms_0$, leading eigenvalues $\L_m$ and leading eigenvectors $\G_m$, respectively, we consider the general Principal Orthogonal Complement Thresholding (POET) Algorithm \ref{algorithm 1} for $\bms_0$ and $\bms_{0u}$. We also consider the general Constrained $l_1$ Minimization for Inverse Matrix Estimation (CLIME) Algorithm \ref{algorithm 2} and Graphic Lasso Estimation (GLASSO) Algotithm \ref{algorithm 3} for $\V_0:=\bms_0^{-1}$ and $\V_{0u}:=\bms_{0u}^{-1}$.
\begin{algorithm}[!ht]
\caption{POET framework for Scatter Matrix Estimation}\label{algorithm 1}
\begin{algorithmic}[1]
    \Require $\hat{\bms}_0$, $\hat{\L}_m$ and $\hat{\G}_m$
    \Ensure $\bms_0^{\tau}$ and $\bms_{0u}^{\tau}$
    \State Compute $\hat{\bms}_{0u}=\hat{\bms}_0-\hat{\G}_m\hat{\L}_m\hat{\G}_m^\T$.
    \State Compute $\hat{\bms}_{0u}^\tau$ via
    \begin{align*}
        \hat{\Sigma}_{0u,ij}^\tau=
        \begin{cases}
            \hat{\Sigma}_{0u,ii}, & i=j, \\
            s_{ij}(\hat{\Sigma}_{0u,ij})\mathbbm{1}(|\hat{\Sigma}_{0u,ij}|\ge\tau_{ij}), & i\ne j,
        \end{cases}
    \end{align*}
    where $s_{ij}(\cdot)$ is the adaptive thresholding rule such that for $x\in\mR$, $s_{ij}(x)=0$ when $|x|\le \tau_{ij}$, otherwise $|s_{ij}(x)-x|\le \tau_{ij}$. Examples of the adaptive thresholding rule include the hard thresholding $s_{ij}(x)=x\mathbbm{1}(|x|\ge\tau_{ij})$, soft thresholding, SCAD and the adaptive lasso (see \cite{AF2001Wavelet,RLZ2009threshold}). The parameter $\tau_{ij}$ is an entry-dependent threshold, which will be determined based on theoretical analysis.
    \State Compute $\hat{\bms}_0^\tau=\hat{\G}_m\hat{\L}_m\hat{\G}_m^\T+\hat{\bms}_{0u}^\tau$.
\end{algorithmic}
\end{algorithm}

\begin{algorithm}[!ht]
\caption{CLIME framework for Inverse Scatter Matrix Estimation}\label{algorithm 2}
\begin{algorithmic}[1]
    \Require $\hat{\bms}_0$, $\hat{\L}_m$ and $\hat{\G}_m$
    \Ensure $\hat{\V}_0$ and $\hat{\V}_{0u}$
    \State Compute $\hat{\bms}_{0u}=\hat{\bms}_0-\hat{\G}_m\hat{\L}_m\hat{\G}_m^\T$.
    \State Compute $\V_{0u}^1$ via 
    \begin{align*}
        \hat{\V}_{0u}^1=\mathop{\arg\min}\limits_{\V}\Vert\V\Vert_{1,1}~~\text{subject to}~~\Vert\hat{\bms}_{0u}\V-\I_d\Vert_{\max}\le\tau,
    \end{align*}
    where $\tau$ is a tuning parameter.
    \State Compute $\hat{\V}_{0u}$ via
    \begin{align*}
        \hat{\V}_{0u,ij}=\hat{\V}_{0u,ij}^1\mathbbm{1}(|\hat{\V}_{0u,ij}^1|\le|\hat{\V}_{0u,ji}^1|)+\hat{\V}_{0u,ji}^1\mathbbm{1}(|\hat{\V}_{0u,ij}^1|>|\hat{\V}_{0u,ji}^1|).
    \end{align*}
    \State Compute $\hat{\V}_0$ via
    \begin{align*}
        \hat{\V}_0=\hat{\V}_{0u}-\hat{\V}_{0u}\hat{\G}_m(\hat{\L}_m^{-1}+\hat{\G}_m^{\T}\hat{\V}_{0u}\hat{\G}_m)^{-1}\hat{\G}_m^{\T}\hat{\V}_{0u}.
    \end{align*}
\end{algorithmic}
\end{algorithm}

\begin{algorithm}[!ht]
\caption{GLASSO framework for Inverse Scatter Matrix Estimation}\label{algorithm 3}
\begin{algorithmic}[1]
    \Require $\hat{\bms}_0$, $\hat{\L}_m$ and $\hat{\G}_m$
    \Ensure $\hat{\V}_0$ and $\hat{\V}_{0u}$
    \State Compute $\hat{\bms}_{0u}=\hat{\bms}_0-\hat{\G}_m\hat{\L}_m\hat{\G}_m^\T$.
    \State Compute $\hat{\V}_{0u}$ via
    \begin{align*}
        \hat{V}_{0u}=\mathop{\arg\min}\limits_{\V=\V^{\T},\V\succ 0}\{\tr(\hat{\bms}_{0u}\V)-\log\det(\V)+\tau\Vert\V\Vert_{1,1}\},
    \end{align*}
    where $\tau$ is a tuning parameter.
    \State Compute $\hat{\V}_0$ via
    \begin{align*}
        \hat{\V}_0=\hat{\V}_{0u}-\hat{\V}_{0u}\hat{\G}_m(\hat{\L}_m^{-1}+\hat{\G}_m^{\T}\hat{\V}_{0u}\hat{\G}_m)^{-1}\hat{\G}_m^{\T}\hat{\V}_{0u}.
    \end{align*}
\end{algorithmic}
\end{algorithm}

\subsection{Spatial-Sign based estimators}
In this section, without assuming the sparsity of the leading vectors $u_j(\bms_0),j\le m$, we provide the specific estimators $\hat{\bms}_0$, $\hat{\L}_m$ and $\hat{\G}_m$. The spatial-sign covariance matrix is defined as $$\bfS=\E\{U(\X-\bmu)U(\X-\bmu)^\T\},$$ where $U(\x)=\x/\Vert\x\Vert_2\mathbbm{1}(\x\ne 0)$ is the spatial sign function. By Theorem 3.2 in \cite{HL2018ECA}, we know that for all $1\le j\le d$, $u_j(\bms_0)=u_j(\bfS)$, and $\lambda_j(\bms_0) \asymp d\lambda_j(\bfS)$ when ${\log d}\Vert\bms\Vert_F=o\{\tr(\bms)\}$. Motivated by this observation, we estimate $\bms_0$ as
\begin{align}\label{scatter matrix estimation}
    \hat{\bms}_0=\frac{d}{n}\sum_{i=1}^nU(\X_i-\hat{\bmu})U(\X_i-\hat{\bmu})^\T,
\end{align}
where $\hat{\bmu}$ is the spatial median, i.e. 
\begin{align*}
    \hat{\bmu}=\mathop{\arg\min}\limits_{\bmu\in\mR^d}\sum_{i=1}^n\Vert\X_i-\bmu\Vert_2.
\end{align*}
Accordingly, we propose 
\begin{align}\label{leading eigen estimation}
    \hat{\L}_m=\diag(\lambda_1(\hat{\bms}_0),\ldots,\lambda_m(\hat{\bms}_0)),~~\text{and}~~ \hat{\G}_m=(u_1(\hat{\bms}_0),\ldots,u_m(\hat{\bms}_0)).
\end{align}

To evaluate the convergence rate of these estimators, we provide the following assumptions. Define $r_i=\Vert \X_i-\bmu\Vert_2$, $\zeta_k=\E(r_i^{-k})$ and $\nu_i=\zeta_1^{-1}r_i^{-1}$.

\begin{assumption}\label{assumption 2}
    There exist constants $0<c<C<\infty$ such that $c<\min_{1\le i\le d}\bms_{ii}\le \max_{1\le i\le d}\bms_{ii}<C$. 
\end{assumption}
\begin{assumption}\label{assumption 3}
    (\romannumeral1) $\zeta_k\zeta_1^{-k}<\zeta\in(0,\infty)$ for $k=1,2,3,4$ and all $d$.\\
    (\romannumeral2) $\nu_i$ is sub-Gaussian distributed, i.e., $\Vert\nu_i\Vert_{\psi_2}\le K_\nu<\infty$.\\
    (\romannumeral3) $\limsup_d\Vert\bfS\Vert_2<1-\psi<1$ for some constant $\psi>0$.
\end{assumption}
Assumption \ref{assumption 3} is widely assumed in high dimensional spatial-sign based procedures, such as \cite{FengJASA2016}, \cite{ZWF2025SSPCA}. By Assumption \ref{assumption 3} (\romannumeral1), we have $\E(\nu_i)=1$, $\var(\nu_i)<\zeta-1$ and $\E(\nu_i^4)<\zeta$. Assumption \ref{assumption 3} (\romannumeral3) means that the maximum eigenvalue of $\bfS$ should be uniformly smaller than one, which is employed to guarantee the consistency of the spatial median. The following theorem establishes the convergence rate of the estimators $\hat{\bms}_0$, $\hat{\L}_m$ and $\hat{\G}_m$.
\begin{theorem}\label{thm 1}
    Under Assumptions \ref{assumption 1}-\ref{assumption 3}, if $\log d=o(n)$, then as $\min(n,d)\to\infty$, 
    \begin{align*}
        &\Vert \hat{\bms}_0-\bms_0\Vert_{\max}=O_p\left(\sqrt{\frac{\log d}{n}}+\sqrt{\frac{\log n}{n}} \right),\\
        &\Vert \hat{\L}_m\L_m^{-1}-\I_m \Vert_{\max}=O_p\left(\sqrt{\frac{\log d}{n}}+\sqrt{\frac{\log n}{n}} \right),~~\text{and}\\
        &\Vert \hat{\G}_m-\G_m\Vert_{\max}=O_p\left(\sqrt{\frac{\log d}{nd}}+\sqrt{\frac{\log n}{nd}} \right).
    \end{align*}
\end{theorem}
Now we apply Algorithm \ref{algorithm 1} to obtain $\hat{\bms}^{\tau}_{0u}$ and $\hat{\bms}_0^{\tau}$. It is worth noting that as the absence of residuals, the standard error estimator of $\hat{\Sigma}_{0u,ij}$ cannot be easily obtained. Thus, in contrast to the choice of $\tau_{ij}$ in \cite{fan2013large}, the threshold parameter is set to be element-wise constant, i.e., $\tau_{ij}=Cw_n$ with a sufficiently large $C>0$ and 
\begin{align*}
    w_n=\sqrt{\frac{\log d}{n}}+\sqrt{\frac{\log n}{n}}.
\end{align*}
Recalling the spectral decomposition of $\bms_0$ given in \eqref{spectral decomposition}, when $d$ is close to or larger than $n$, estimating $\bms_{0u}$ is very challenging. Therefore, following \cite{fan2013large}, we consider the sparsity structure of $\bms_{0u}$ as follows
\begin{align*}
    m_d=\max_{1\le i\le d}\sum_{j=1}^d|\bms_{0u,ij}|^v~~\text{for some}~~v\in[0,1],
\end{align*}
which will subsequently be assumed to decay at a certain rate. The next theorem establishes the asymptotic properties of the thresholding estimators $\hat{\bms}_{0u}^{\tau}$ and $\hat{\bms}_{0}^{\tau}$.
\begin{theorem}\label{thm 2}
    Under Assumption \ref{assumption 1}-\ref{assumption 3}, if $\log d=o(n)$ and $w_n^{1-v}m_d=o(1)$, then as $\min(n,d)\to\infty$,
    \begin{align*}
        &\Vert\hat{\bms}_{0u}^{\tau}-\bms_{0u}\Vert_2=O_p(w_n^{1-v}m_d)=\Vert(\hat{\bms}_{0u}^{\tau})^{-1}-\bms_{0u}^{-1}\Vert_2,\\
        &\Vert\hat{\bms}_0^{\tau}-\bms_0\Vert_{\bms_0}=O_p(d^{1/2}w_n^2+w_n^{1-v}m_d),~~\text{and}\\
        &\Vert(\hat{\bms}_0^{\tau})^{-1}-\bms_0^{-1}\Vert_2=O_p(w_n^{1-v}m_d).
    \end{align*}
\end{theorem}

In many cases, making the sparsity assumption on $\V_{0u}=\bms_{0u}^{-1}$ appears more natural compared to the sparsity of $\bms_{0u}$. For example, for $\y\sim EC_d(0,\bms_{0u},\xi)$, the sparsity of $\V_{0u}=\bms_{0u}^{-1}$ encodes the conditional uncorrelatedness relationships between all variables in $\y$. Therefore, following \cite{CaiLiuLuo2011}, we require the sparsity of $\V_{0u}$, which is measured by
\begin{align*}
    M_d=\max_{1\le i\le d}\sum_{j=1}^d|\V_{0u,ij}|^v~~\text{for some}~~v\in[0,1]
\end{align*}
Then we apply Algorithm \ref{algorithm 2} or \ref{algorithm 3} to obtain $\hat{\V}_{0u}$ and $\hat{\V}_0$. The tuning parameter is set to be $\tau=Cw_n$ for some large enough $C>0$.
 
\begin{assumption}\label{assumption 4}
    (For CLIME) $\Vert\V_{0u}\Vert_{\infty}=O(1)$.\\
    (For GLASSO) Let $\kappa\triangleq\max_{1\le i\le d}|\{j:V_{0u,ij}\ne 0\}|$, $\mathcal{E}\triangleq\{(i,j)\in[d]\times[d]:i\ne j,V_{0u,ij}\ne 0\}$, $\mathcal{S}\triangleq\mathcal{E}\cup\{(1,1),\ldots,(d,d)\}$ and $\bms_{0u}^*\triangleq\bms_{0u}\otimes\bms_{0u}$. We suppose \\
    (\romannumeral1) $\kappa=O(1)$, $|\mathcal{E}|=O(d)$.\\
    (\romannumeral2) $\Vert\bms_{0u,\mathcal{S}^c\mathcal{S}}^*\bms_{0u,\mathcal{S}\mathcal{S}}^*\Vert_{\infty}\le 1-\alpha$ for some $\alpha\in(0,1)$.\\
    (\romannumeral3) $\max(\Vert(\bms_{0u,\mathcal{S}\mathcal{S}}^*)^{-1}\Vert_{\infty},\Vert\bms_{0u}\Vert_{\infty})=O(1)$.
\end{assumption}
Assumption \ref{assumption 4} (for GLASSO) aligns with the conditions assumed in \cite{Ravikumar2011}. The following theorem gives the rates of convergence for $\hat{\V}_{0u}$ and $\hat{\V}_0$.
\begin{theorem}\label{thm 3}
    Under Assumption \ref{assumption 1}-\ref{assumption 4}, if $\log d=o(n)$ and $w_n^{1-v}M_d=o(1)$, then as $\min(n,d)\to\infty$,
    \begin{align*}
        &\Vert\hat{\V}_{0u}-\V_{0u}\Vert_{\max}=O_p(w_n)=\Vert\hat{\V}_0-\V_0\Vert_{\max},~~\text{and}\\
        &\Vert\hat{\V}_{0u}-\V_{0u}\Vert_2=O_p(w_n^{1-v}M_d)=\Vert\hat{\V}_0-\V_0\Vert_2.
    \end{align*}
\end{theorem}

\subsection{The number of spiked eigenvalues estimation}\label{secnum}

Recalling the spiked structure of $\bms_0$ given in Assumption \ref{assumption 1}, by Theorem \ref{thm 1} and Wely's Theorem, we have
\begin{align*}
    \lambda_j(\hat{\bms}_0)=
    \begin{cases}
        c_jd,&~~\text{if}~~1\le j\le m,\\
        c_jdw_n,&~~\text{if}~~m+1\le j\le d,
    \end{cases}
\end{align*}
with probability tending to one, where $c_j$'s are some positive constants. Then we estimate the number of spiked eigenvalues $m$ via the following eigenvalue ratio or growth ratio,
\begin{align}\label{spike number estimation}
    \hat{m}_{ER}=\mathop{\arg\max}\limits_{1\le j\le M}\frac{\lambda_j(\hat{\bms}_0)}{\lambda_{j+1}(\hat{\bms}_0)},~~\hat{m}_{GR}=\mathop{\arg\max}\limits_{1\le j\le M}\frac{\ln\{1+\lambda_j(\hat{\bms}_0)/V_{j-1}\}}{\ln\{1+\lambda_{j+1}(\hat{\bms}_0)/V_j\}},
\end{align}
where $M$ is a predetermined upper bound of the true $m$, and $V_j=\sum_{l=j+1}^{\min(n,d)-1}\lambda_l(\hat{\bms}_0)$. \cite{Ahn20131203} recommended two methods to choose a suitable value for $M$, which are also suitable for our estimators. In the simulation study, $M$ is set as $8$ but this can be replaced with any other reasonable values. The following theorem show the consistency of $\hat{m}_{ER}$ and $\hat{m}_{GR}$.
\begin{theorem}\label{thm 4}
    Suppose the conditions in Theorem \ref{thm 1} hold, then as $\min(n,d)\to\infty$, for $m\le M \le\min(n,d)-1$, $\Pr(\hat{m}_{ER}=m)\to 1$ and $\Pr(\hat{m}_{GR}=m)\to 1$.
\end{theorem}

\section{Large Matrix Estimation based on Tyler's M-estimator}
Although the high-dimensional spatial-sign covariance matrix can serve as a good approximation to the true scatter matrix under mild conditions, a  theoretical bias still remains between the two. In contrast, as shown in \cite{Tyler1987}, the classical Tyler’s M-estimator coincides exactly with the scatter matrix. Motivated by this property, we also extend the POET framework to incorporate Tyler’s M-estimator.

If the location parameter is known, \cite{Tyler1987} showed that the scatter matrix 
$\hat{\boldsymbol{\Sigma}}$ can be estimated by solving the nonlinear equation
\begin{align*}
\frac{d}{n}\sum_{i=1}^n 
\frac{(\mathbf{X}_i-\boldsymbol{\mu})(\mathbf{X}_i-\boldsymbol{\mu})^{\T}}
{(\mathbf{X}_i-\boldsymbol{\mu})^{\T}
 \hat{\boldsymbol{\Sigma}}^{-1}
 (\mathbf{X}_i-\boldsymbol{\mu})}
= \hat{\boldsymbol{\Sigma}} .
\end{align*}
Let $\boldsymbol{x}_i=\mathbf{X}_i-\boldsymbol{\mu}$.  
\cite{Tyler1987} suggested solving this equation iteratively via
\begin{align*}
\hat{\boldsymbol{\Sigma}}_{k+1}
= 
\frac{
d \sum_{i=1}^n 
\frac{\boldsymbol{x}_i \boldsymbol{x}_i^\T}
{\boldsymbol{x}_i^\T \hat{\boldsymbol{\Sigma}}_k^{-1} \boldsymbol{x}_i}
}{
\operatorname{tr}\!\left(
\sum_{i=1}^n 
\frac{\boldsymbol{x}_i \boldsymbol{x}_i^\T}
{\boldsymbol{x}_i^\T \hat{\boldsymbol{\Sigma}}_k^{-1} \boldsymbol{x}_i}
\right)
}.
\end{align*}
\cite{KentTyler1991} proved that, under mild regularity conditions, this algorithm admits a unique solution and converges globally. The resulting Tyler's M-estimator (TME) is also the maximum likelihood estimator of the shape matrix for the angular central Gaussian distribution \citep{Tyler1987b} and for generalized elliptical distributions \citep{FRAHM2010374}.

However, when $d>n$, the classical TME becomes ill-posed. To address this issue, \cite{sun2014regularized} considered the regularized TME $\hat{\boldsymbol{\Sigma}}(\alpha)$, defined as the solution to
\begin{align*}
\hat{\boldsymbol{\Sigma}}(\alpha)
=
\frac{1}{1+\alpha}
\frac{d}{n}
\sum_{i=1}^n
\frac{\boldsymbol{x}_i \boldsymbol{x}_i^\T}
{\boldsymbol{x}_i^\T \hat{\boldsymbol{\Sigma}}(\alpha)^{-1} \boldsymbol{x}_i}
+
\frac{\alpha}{1+\alpha}\mathbf{I}_d ,
\end{align*}
where $\alpha>0$ is a regularization parameter. They further proposed the iterative algorithm
\begin{align*}
\hat{\boldsymbol{\Sigma}}_{k+1}(\alpha)
=
\frac{1}{1+\alpha}
\frac{d}{n}
\sum_{i=1}^n
\frac{\boldsymbol{x}_i \boldsymbol{x}_i^\T}
{\boldsymbol{x}_i^\T \hat{\boldsymbol{\Sigma}}_k(\alpha)^{-1} \boldsymbol{x}_i}
+
\frac{\alpha}{1+\alpha}\mathbf{I}_d ,
\end{align*}
which converges rapidly in practice. Based on this estimator, \cite{goes2020robust} proposed a thresholding estimator for sparse scatter matrices. Although these approaches are computationally efficient, they fundamentally rely on sparsity assumptions, which may lead to inefficiency when the true scatter matrix is nonsparse. Motivated by this limitation, we develop new POET-type estimators based on TME that remain effective in nonsparse settings.

By the Tyler's M-estimator matrix definition, we proposed a new naive estimator of the scatter matrix $\bms_0$ as follows
\begin{align} \label{tmeeq}
\hat{\boldsymbol{\Sigma}}_T=\frac{d}{n}\sum_{i=1}^n\frac{(\X_i-\hat{\bmu})(\X_i-\hat{\bmu})^{\T}}{(\X_i-\hat{\bmu})^{\T} \hat{\V}_{S}(\X_i-\hat{\bmu})}
\end{align}
where $\hat{\V}_{S}$ is the above POET estimators of $\V_0=\bms_0^{-1}$ based on spatial-sign covariance matrix, i.e., $\hat{\V}_{S}=(\hat{\bms}_{0}^{\tau})^{-1}$ or $\hat{\V}_0$. Theorem \ref{thm 2} and \ref{thm 3} yield
\begin{align*}
    \Vert\hat{\V}_{S}-\V_0\Vert_2=O_p(a_n)~~\text{with}~~a_n:=
    \begin{cases}
        w_n^{1-v}m_d,&~~\text{if}~~\hat{\V}_{S}=(\hat{\bms}_{0}^{\tau})^{-1},\\
        w_n^{1-v}M_d,&~~\text{if}~~\hat{\V}_{S}=\hat{\V}_0.
    \end{cases}
\end{align*}
To proceed, we also adopt the POET procedure for the naive Tyler's estimator $\hat{\bms}_T$ to obtain the final estimator $\hat{\bms}_{TP}$ for $\bms_0$, or adopt the CLIME or GLASSO procedures to obtain the final estimator $\hat{\V}_{T}$ for $\V_0$. The next theorem establishes the convergence rate of the POET-TME estimator $\hat{\bms}_{TP}$ and CLIME-TME or GLASSO-TME estimators $\hat{\V}_{T}$.
\begin{theorem}\label{thm 7}
    Under Assumption \ref{assumption 1}-\ref{assumption 3} and the assumption that $\log d=o(n)$, if $\Vert\hat{\V}_{S}-\V_0\Vert_2=O_p(a_n)$ and $a_n=O(w_n)$, then
    \begin{align*}
        &\Vert\hat{\bms}_T-\bms_0\Vert_{\max}=O_p(w_n),\\
        &\Vert\hat{\bms}_{TP}-\bms_0\Vert_{\bms_0}=O_p(d^{1/2}w_n^2+w_n^{1-v}m_d),\\
        &\Vert(\hat{\bms}_{TP})^{-1}-\bms_0^{-1}\Vert_2=O_p(w_n^{1-v}m_d).
        \end{align*}
Additionally, if Assumption \ref{assumption 4} also holds, we have
        \begin{align*}
        \Vert\hat{\V}_T-\V_0\Vert_{\max}=O_p(w_n),
        \Vert\hat{\V}_T-\V_0\Vert_2=O_p(w_n^{1-v}M_d).
    \end{align*}
\end{theorem}

\begin{remark}
A natural idea is to iterate the estimating equation~\eqref{tmeeq} by replacing it with
\begin{align} \label{tmeeq2}
\hat{\bms}_{TP,k+1}^{(0)}
= \frac{d}{n}\sum_{i=1}^n 
\frac{(\X_i-\hat{\bmu})(\X_i-\hat{\bmu})^{\T}}
     {(\X_i-\hat{\bmu})^{\T} \hat{\bms}_{TP,k}^{-1}(\X_i-\hat{\bmu})},
\end{align}
where $\hat{\bms}_{TP,k}$ denotes the estimator obtained at the $k$-th step. 
After computing the intermediate estimator $\hat{\bms}_{TP,k+1}^{(0)}$, 
one may further apply the POET procedure to obtain the refined estimator 
$\hat{\bms}_{TP,k+1}$. Repeating these two steps yields an iterative algorithm 
that in principle can be continued until convergence.

However, our simulation studies indicate that this fully iterative scheme 
does not provide substantial improvement over the proposed POET--TME estimator. 
In fact, the estimator obtained after the first POET refinement already achieves 
a level of accuracy comparable to that of the fully iterated sequence. 
This suggests that the primary gain comes from the initial POET correction, 
and additional iterations produce only marginal benefits while incurring extra 
computational cost. Therefore, in practice, performing this replication step once 
is sufficient to obtain a statistically efficient and computationally scalable estimator.
\end{remark}

\begin{remark}
\cite{ding2025sub} develop an idiosyncratic-projected self-normalization 
method for estimating the scatter matrix. 
Intuitively, their approach replaces the initial estimator $\V_S$ in~\eqref{tmeeq} 
with the matrix $\hat{\mathbf{P}}^{\T}\hat{\mathbf{P}}$, where $\hat{\mathbf{P}}$ is a $(d-m)\times d$ 
projection matrix satisfying 
$\hat{\mathbf{P}}\hat{\G}_{KED}=0$ and $\hat{\mathbf{P}}\hat{\mathbf{P}}^{\T}=\I_{d-m}$. 
Here, $\hat{\G}_{KED}$ denotes the $m$ leading eigenvectors of the 
multivariate Kendall's tau matrix. 
In essence, the estimator $\hat{\mathbf{P}}^{\T}\hat{\mathbf{P}}$ is used as a surrogate for 
the precision matrix $\bms^{-1}$. 
However, this construction focuses solely on the leading $m$ principal components 
and consequently discards the information contained in the residual component 
$\bms_u$. 

As shown in Theorems~\ref{thm 2} and~\ref{thm 3}, the POET-based estimator 
constructed from the spatial-sign covariance matrix provides a more accurate 
and stable initial estimator for the scatter matrix, as it retains information 
from both the factor structure and the idiosyncratic component. 
This highlights the importance of choosing a high-quality initial estimator in 
the Tyler-type fixed-point iteration~\eqref{tmeeq}. 
Moreover, the same principle suggests that other well-performing scatter estimators 
may also serve as effective substitutes for $\hat{\V}_S$ in~\eqref{tmeeq}, 
leading to procedures analogous to our proposed method. 
Exploring such alternatives may further broaden the applicability and robustness 
of POET-enhanced Tyler estimators in high-dimensional heavy-tailed settings.

In addition, \cite{ding2025sub} employ the Huber estimator to estimate the 
location parameter $\bmu$. Although the Huber estimator is well known for its 
robustness under heavy-tailed distributions, it is not the most natural choice 
for elliptically symmetric distributions. In such settings, the sample spatial 
median is widely regarded as a more suitable and efficient estimator of location; 
see \citep{oja2010multivariate}. 
The spatial median enjoys affine equivariance and high robustness, and it aligns 
well with the geometric structure of elliptical distributions. 
Therefore, when the underlying model adheres to elliptical symmetry, replacing the 
Huber estimator with the sample spatial median may lead to improved statistical 
efficiency and better compatibility with scatter estimation procedures based on 
self-normalization or spatial-sign transformations.

\end{remark}

Furthermore, we also proposed an POET-TME estimator of the covariance matrix $\bms_{\X}:=\cov(\X_i)$. In detail, note that
\begin{align*}
    \bms_{\X}=d^{-1}\E(\xi_i^2)\bms=d^{-1}\E(r_{T,i}^2)\bms_0,~~\text{where}~~r_{T,i}:=\Vert\V_0^{1/2}(\X_i-\bmu)\Vert_2.
\end{align*}
We estimate the $d^{-1}\E(r_{T,i}^2)$ by the following robust estimator $d^{-1}\widehat{\E(r_{T,i}^2)}$ that solves
\begin{align*}
    \sum_{i=1}^nH_h\left(d^{-1}\hat{r}_{T,i}^2-d^{-1}\widehat{\E(r_{T,i}^2)}\right)=0~~\text{where}~~\hat{r}_{T,i}:=\Vert\hat{\V}_S^{1/2}(\X_i-\hat{\bmu})\Vert_2,
\end{align*}
$H_h(x)=\min\{h,\max(-h,x)\}$ is the Huber function, and $h$ is a tuning parameter. The estimator of $\bms_{\X}$ is then simply
\begin{align*}
    \hat{\bms}_{\X,T}=d^{-1}\widehat{\E(r_{T,i}^2)}\hat{\bms}_T,~~\text{and}~~\hat{\bms}_{\X,TP}=d^{-1}\widehat{\E(r_{T,i}^2)}\hat{\bms}_{TP}.
\end{align*}
The next proposition gives the convergence rate of $d^{-1}\widehat{\E(r_{T,i}^2)}$ and $\hat{\bms}_{\X,T},\hat{\bms}_{\X,TP}$.
\begin{proposition}
    Suppose the conditions in Theorem \ref{thm 7} hold, in addition, $\E(|\xi/\sqrt{d}|^{2+\varepsilon})<\infty$ for some $\varepsilon\in(0,2)$, and $h=cn^{2/(2+\varepsilon)}$ for some constant $c>0$. Then,
    \begin{align*}
        &\left|d^{-1}\widehat{\E(r_{T,i}^2)}-d^{-1}\E(r_{T,i}^2)\right|=O_p(w_n+n^{-\frac{\varepsilon}{2+\varepsilon}}),\\
        &\Vert\hat{\bms}_{\X,T}-\bms_{\X}\Vert_{\max}=O_p(w_n+n^{-\frac{\varepsilon}{2+\varepsilon}}),\\
        &\Vert\hat{\bms}_{\X,TP}-\bms_{\X}\Vert_{\bms_{\X}}=O_p(d^{1/2}w_n^2+w_n^{1-v}m_d+n^{-\frac{\varepsilon}{2+\varepsilon}}),~~\text{and}\\
        &\Vert(\hat{\bms}_{\X,TP})^{-1}-\bms_{\X}^{-1}\Vert_2=O_p(w_n^{1-v}m_d+n^{-\frac{\varepsilon}{2+\varepsilon}}).
    \end{align*}
\end{proposition}

\section{Simulation}
\subsection{Estimation of scatter matrix and covariance matrix}\label{sec41}
We evaluate the finite-sample performance of the proposed POET-SS and POET-TME estimators under high-dimensional elliptical factor models.
Given \( n \) independent samples \((\f_j, \u_j)_{1\le j\le n}\) of \((\f, \u)\), the observed data are constructed according to the model $ \y_j = \B \f_j + \u_j$.
The population covariance matrix of \( \y \) is therefore given by $\bms_0 = \frac{d(\B\B^{\T} + \mathbf{\Sigma}_u)}{\tr(\B\B^{\T} + \mathbf{\Sigma}_u)}$.
In the non-sparse setting, we consider a loading matrix $\B = (B_{ij})_{1 \leq i \leq d,\, 1 \leq j \leq m}$, where the entries are independently generated as $ B_{ij} \stackrel{\text{i.i.d.}}{\sim} N(0, s_i) $. We set the number of factors to \( m = 3 \) and specify the loading variances as \( s_1 = 1 \), \( s_2 = 0.75^2 \), and \( s_3 = 0.5^2 \). For a fair comparison, we set the number of factors are all known for all methods. To simulate the data, we generate pairs \((\f, \u)\) from the following distributions:
\begin{itemize}
\item[(I)] multivariate normal distribution $N(\bm 0,\mathbf{\Sigma}_{fu})$;
\item[(II)] multivariate $t$-distribution $t_\nu(\bm 0,\mathbf{\Sigma}_{fu})$ with degree of freedom $\nu=4$;
\item[(III)] multivariate $t$-distribution $t_\nu(\bm 0,\mathbf{\Sigma}_{fu})$ with degree of freedom $\nu=2.2$;
\item[(IV)] mixture normal distribution $0.8N(\bm 0,\mathbf{\Sigma}_{fu})+0.2N(\bm 0,10\mathbf{\Sigma}_{fu})$.
\end{itemize} 
Here covariance matrix \( \mathbf{\Sigma}_{fu}=c \cdot \diag(\I_m,\mathbf{\Sigma}_u) \), where the normalization constant is defined as $c = \frac{d}{\tr(\B\B^{\T} + \mathbf{\Sigma}_u)}$.   Here we set $\mathbf{\Sigma}_u=(0.9^{|i-j|})_{1\le i,j\le d}$.

We assess the performance of five scatter matrix estimators: 
\begin{itemize}
\item the proposed POET estimator based on spatial sign covariance, denoted by POET-SS; 
\item the proposed POET estimator based on Tyler's M-estimator, denoted by POET-TME
\item the robust generic POET estimator from \cite{FLW2018generalPOET}, denoted by FLW; 
\item the original POET estimator based on the sample covariance matrix \citep{fan2013large}, denoted by SAMPLE;
\item the POET estimator applied to the regularized Tyler’s M-estimator (TME) from \cite{goes2020robust}, denoted by RegTME.
\end{itemize}

To implement RegTME, we first symmetrize the data to remove the mean via $\tilde{\y}_j = \y_{2j+1} - \y_{2j+2}$, $\text{for } j = 1, \ldots, \left\lfloor \frac{n}{2} \right\rfloor$, and then apply the regularized TME to $\tilde{\y}_j$. Following Theorems 1–2 in \citet{goes2020robust}, the regularization parameter $\alpha$ in equations (1.4) and (1.5) is set as $\alpha = \max\left(0.1,\ 1.1\left(\gamma-1 + \hat{s}_{\max}(1 + \sqrt{\gamma})^2\right)\right)$, where $\gamma = \frac{d}{2n}$ and $\hat{s}_{\max}$ denotes the spectral norm of the sample covariance matrix of $\tilde{\y}_j$. For the FLW and SAMPLE methods, $\hat{\bms}$ and $\hat{\L}_m$ are normalized by $d/\tr(\hat{\bms})$ to obtain estimators of $\bms_0$ and $\L_{m}$. The maximum norm error for estimating $\Gamma_m$ is scaled by $\sqrt{d}$, i.e.,$
\sqrt{d}\Vert\hat{\G}_m - \G_m \Vert_{\max} $. For $\L_{m}$, we report the ratio error: $
\Vert \hat{\L}_{m} (\L_{m})^{-1} - \I_m\Vert_{\max}$.

\begin{table}[htbp]
  \centering
  \caption{Estimation performance of each method for scatter matrix with $(n,d)=(250,500)$. Reported values are the mean and standardized deviation (in parentheses).}
  \label{tab:table1}
  \setlength{\tabcolsep}{2pt} 
  \begin{tabular}{llllll}
    \toprule
 & {\small$\|\hat{\bms}_0-\bms_0\|_{max}$} 
 & {\small $\|\hat{\L}_K\L^{-1}_{K}-\I_d\|_{max}$} 
 & {\small$\sqrt{d}\|\hat{\G}_K-\G_{K}\|_{max}$}  
 & {\small$\|\hat{\bms}_0-\bms_0\|_{\bms_0}$ } 
 & {\small$\|\hat{\bms}_{0u}-\bms_{0u}\|_2$} \\   \hline
\multicolumn{6}{c}{Scenario (I)} \\ \hline
POET-SS   & 0.927 (0.234) & 0.228 (0.041) & 0.589 (0.173) & 3.97 (0.149) & \textbf{3.41} (0.184) \\
POET-TME  & \textbf{0.485} (0.116) & \textbf{0.094} (0.046) & 0.536 (0.155) & 2.97 (0.139) & 4.04 (0.164) \\
IPSN      & 0.512 (0.128) & 0.103 (0.048) & 0.542 (0.159) & 3.11 (0.150) & 4.38 (0.158) \\
FLW       & 0.497 (0.117) & 0.106 (0.044) & 0.547 (0.158) & 4.06 (0.156) & 3.95 (0.172) \\
SAMPLE    & 0.488 (0.121) & 0.095 (0.045) & \textbf{0.534} (0.157) & \textbf{2.46} (0.121) & 3.69 (0.173) \\
RegTME    & 1.01 (0.254) & 0.229 (0.055) & 0.853 (0.289) & 4.34 (0.221) & 3.76 (0.212) \\ \hline

\multicolumn{6}{c}{Scenario (II)} \\ \hline
POET-SS   & 0.950 (0.238) & 0.231 (0.042) & 0.601 (0.179) & 3.97 (0.154)  & \textbf{3.43} (0.191) \\
POET-TME  & \textbf{0.489} (0.126) & \textbf{0.098} (0.051) & \textbf{0.539} (0.154) & \textbf{2.97} (0.146)  & 4.05 (0.171) \\
IPSN      & 0.513 (0.133) & 0.107 (0.053) & 0.548 (0.158) & 3.12 (0.160)  & 4.39 (0.164) \\
FLW       & 0.785 (0.392) & 0.145 (0.055) & 0.589 (0.168) & 6.39 (1.540)   & 4.48 (1.600) \\
SAMPLE    & 0.919 (0.475) & 0.181 (0.112) & 0.990 (0.438) & 2.95 (0.953)  & 4.09 (0.939) \\
RegTME    & 1.030 (0.277)  & 0.230 (0.056) & 0.846 (0.255) & 4.33 (0.227)  & 3.78 (0.220) \\ \hline

\multicolumn{6}{c}{Scenario (III)} \\ \hline
POET-SS & 0.950 (0.233) & 0.231 (0.0415) & 0.600 (0.189) & 3.98 (0.151) & \textbf{3.42} (0.186) \\
POET-TME & \textbf{0.491} (0.124) & \textbf{0.101} (0.052) & \textbf{0.540} (0.166) & \textbf{2.97} (0.141) & 4.04 (0.167) \\
IPSN & 0.519 (0.133) & 0.109 (0.055) & 0.552 (0.178) & 3.14 (0.156) & 4.38 (0.161) \\
FLW & 1.79 (1.20) & 0.214 (0.077) & 0.619 (0.193) & 11.97 (5.35) & 9.41 (7.03) \\
SAMPLE & 2.23 (1.57) & 0.381 (0.296) & 1.96 (0.707) & 6.96 (5.27) & 7.21 (3.91) \\
RegTME & 1.02 (0.264) & 0.232 (0.059) & 0.869 (0.318) & 4.34 (0.217) & 3.76 (0.216) \\ \hline

\multicolumn{6}{c}{Scenario (IV)} \\ \hline
POET-SS   & 0.944 (0.249) & 0.229 (0.040) & 0.597 (0.188) & 3.97 (0.152) & \textbf{3.42} (0.186) \\
POET-TME  & \textbf{0.489} (0.129) & \textbf{0.096} (0.048) & \textbf{0.537} (0.160) & \textbf{2.97} (0.139) & 4.05 (0.165) \\
IPSN      & 0.594 (0.205) & 0.109 (0.055) & 0.594 (0.191) & 3.35 (0.172) & 4.41 (0.163) \\
FLW       & 0.818 (0.218) & 0.183 (0.053) & 0.638 (0.190) & 8.28 (0.400) & 4.46 (0.220) \\
SAMPLE    & 1.020  (0.271) & 0.197 (0.100) & 1.190  (0.405) & 3.60 (0.390) & 4.29 (0.282) \\
RegTME    & 1.020  (0.271) & 0.229 (0.054) & 0.851 (0.292) & 4.34 (0.219) & 3.77 (0.221) \\ \hline
  \end{tabular}
\end{table}

Table~\ref{tab:table1} reports the estimation results for all competing methods under the setting $(n,d)=(250,500)$. Overall, the proposed POET-TME procedure demonstrates the most stable and accurate performance across a wide range of distributions.

For the Gaussian case, both SAMPLE and POET-TME perform comparably in terms of estimation accuracy, and SAMPLE enjoys high computational efficiency. Nevertheless, POET-TME achieves nearly identical accuracy while retaining robustness, indicating that POET-TME does not sacrifice efficiency in well-behaved scenarios.

For the heavy-tailed settings, the superiority of POET-TME becomes much more pronounced. In almost all scenarios, POET-TME yields the smallest estimation errors, particularly in estimating the leading eigenvalues and eigenvectors. This highlights its ability to effectively mitigate the influence of extreme observations and model the underlying low-rank structure even when the data deviates significantly from normality.

Because of the intrinsic bias between the spatial-sign covariance matrix and the true covariance matrix, the POET-SS method exhibits a larger error in estimating $\|\hat{\L}_K\L_K^{-1}-\I_d\|{\max}$ compared to POET-TME and IPSN. However, POET-SS still performs well in estimating the idiosyncratic component and typically attains the smallest $\|\hat{\bms}_{0u}-\bms_{0u}\|_2$.

The RegTME method also suffers from bias in eigenvalue estimation under the nonsparse structure, leading to inferior performance relative to POET-TME. In contrast, POET-TME effectively incorporates thresholding and robust eigenstructure estimation, enabling it to capture both the global low-rank factor structure and the local idiosyncratic variations with high accuracy. As a result, POET-TME consistently achieves a favorable balance between robustness, efficiency, and precision, establishing it as the most reliable estimator among all considered methods.

Next, we evaluate the performance of the covariance matrix estimators. In this part, we focus on four representative methods: POET-TME, SAMPLE, FLW, and IPSN, while keeping all simulation settings identical to those described earlier. The numerical results are summarized in Table~\ref{tab:sigma_err}.

For the multivariate normal distribution, all four estimators exhibit comparable performance across the considered error metrics, indicating that the benefits of robustification are less pronounced under light-tailed settings. In contrast, for the heavy-tailed scenarios (Models (II)–(IV)), substantial differences emerge. POET-TME consistently achieves the smallest estimation errors, demonstrating its superior robustness against tail heaviness. IPSN ranks second and also shows stable performance, followed by FLW, which performs moderately well but is more sensitive to tail behavior. SAMPLE exhibits the largest errors under heavy-tailed distributions, confirming that the sample covariance is unreliable in such settings.

Overall, these results highlight the advantages of robust covariance estimation in heavy-tailed environments and underscore the effectiveness of the POET-TME approach in maintaining accuracy across a wide range of distributional conditions.

\begin{table}[htbp]
\centering
\caption{Estimation performance of each method for covariance matrix with $(n,d)=(250,500)$. Reported values are the mean and standardized deviation (in parentheses).}
\renewcommand{\arraystretch}{1.2}
\begin{tabular}{lccc}
\hline
Method    & $\|\hat{\bms}-\bms\|_{max}$ & $\|\hat{\bms}-\bms\|_{\bms}$ & $\|\hat{\bms}-\bms\|_2$ \\
\hline
\multicolumn{4}{c}{Scenario I}\\ \hline
SAMPLE    & \textbf{0.55}(0.18) & \textbf{2.43}(0.03) & \textbf{24.59}(7.75) \\
POET-TME & 0.56(0.18) & 2.88(0.04) & 24.61(7.33) \\
FLW       & 0.57(0.18) & 4.21(0.12) & 24.96(7.57) \\
IPSN      & 0.57(0.18) & 3.09(0.05) & 25.15(7.98) \\
\hline
\multicolumn{4}{c}{Scenario II}\\ \hline
SAMPLE    & 1.20(1.07) & 3.00(1.94) & 53.68(50.26) \\
POET-TME & \textbf{0.70}(0.29) & \textbf{2.77}(0.33) & \textbf{29.61}(10.58) \\
FLW       & 1.18(1.15) & 7.14(4.05) & 36.19(23.98) \\
IPSN      & 0.72(0.28) & 2.82(0.34) & 29.97(10.63) \\
\hline
\multicolumn{4}{c}{Scenario III}\\ \hline
SAMPLE    & 4.11(11.37) & 6.02(20.91) & 173.08(465.20) \\
POET-TME & \textbf{2.68}(0.68)  & \textbf{0.82}(0.35)  & \textbf{105.35}(19.19)  \\
FLW       & 3.92(11.68) & 12.56(48.41) & 113.01(161.69) \\
IPSN      & 2.79(0.66)  & 0.84(0.30)  & 109.25(17.23)  \\ 
\hline
\multicolumn{4}{c}{Scenario IV}\\ \hline
SAMPLE    & 1.32(0.53) & 3.13(0.30) & 57.73(20.37) \\
POET-TME & \textbf{0.73}(0.33) & \textbf{2.90}(0.43) & \textbf{31.55}(12.58) \\
FLW       & 1.27(0.62) & 9.24(1.40) & 39.95(17.49) \\
IPSN      & 0.89(0.43) & 3.18(0.61) & 36.09(15.28) \\
\hline
\end{tabular}
\label{tab:sigma_err}
\end{table}

\subsection{Estimation of precision matrix}
Next, we turn to the estimation of the precision matrix. The simulation design is identical to that in Section~\ref{sec41}, except that we specify the idiosyncratic precision matrix as
$
\V_u = \left(0.4^{|i-j|}\right)_{1\le i,j\le d}.
$
Table~\ref{tab22} reports various matrix–norm errors for estimating both $\V$ and $\V_u$ across the considered methods. Several clear patterns emerge. First, POET-SS and POET-TME exhibit almost indistinguishable performance under all norms, confirming that incorporating robust scatter estimators into the POET framework yields stable and reliable precision estimation. Second, both methods substantially outperform the alternative approaches when the data are heavy-tailed. In particular, they attain significantly smaller Frobenius, max, and spectral norm errors, highlighting the advantage of leveraging robustness in high-dimensional precision matrix estimation. 

Overall, these results demonstrate that POET-SS and POET-TME provide a powerful and accurate strategy for estimating large precision matrices in the presence of heavy-tailed noise. Moreover, if computational efficiency is also a concern, the POET-SS estimator may be especially appealing due to its lower runtime while maintaining comparable accuracy.

\begin{table}[htbp]\label{tab22}
\centering
  \setlength{\tabcolsep}{2pt} 
\caption{Estimation performance of each method for precision matrix with $(n,d)=(250,500)$. Reported values are the mean and standardized deviation (in parentheses).}
\begin{tabular}{lcccccc}
\hline
Method & $\|\hat{\V}-\V\|_F$ & $\|\hat{\V}_u-\V_u\|_F$ & $\|\hat{\V}-\V\|_{max}$ & $\|\hat{\V}_u-\V_u\|_{max}$ & $\|\hat{\V}-\V\|_2$ & $\|\hat{\V}_u-\V_u\|_2$ \\
\hline
&\multicolumn{6}{c}{Scenario (I)} \\ \hline
POET-SS        
& 2.13 (0.02) & 2.15 (0.02) & \textbf{1.50} (0.15) & \textbf{1.52} (0.15) & 4.27 (0.07) & 4.29 (0.07) \\

POET-TME   
& \textbf{2.11} (0.02) & \textbf{2.12} (0.02) & 1.51 (0.15) & 1.53 (0.15) & \textbf{4.24} (0.07) & \textbf{4.26} (0.07) \\

IPSN        
& 2.39 (0.03) & 2.40 (0.03) & 2.04 (0.17) & 2.06 (0.17) & 4.96 (0.08) & 4.97 (0.08) \\

SAMPLE      
& 3.28 (0.04) & 3.30 (0.04) & 3.01 (0.19) & 3.04 (0.19) & 6.81 (0.09) & 6.83 (0.09) \\

FLW         
& 3.14 (0.03) & 3.16 (0.03) & 2.62 (0.17) & 2.64 (0.17) & 6.49 (0.08) & 6.51 (0.08) \\

RegTME      
& 5.26 (0.05) & 5.28 (0.05) & 5.40 (0.26) & 5.44 (0.25) & 10.3 (0.09) & 10.3 (0.09) \\
\hline
&\multicolumn{6}{c}{Scenario (II)} \\ \hline
POET-SS        
& 2.14 (0.02) & 2.15 (0.02) & \textbf{1.51} (0.14) & \textbf{1.52} (0.15) & 4.28 (0.06) & 4.29 (0.06) \\

POET-TME   
& \textbf{2.11} (0.02) & \textbf{2.12} (0.02) & 1.52 (0.14) & 1.54 (0.15) & \textbf{4.25} (0.06) & \textbf{4.26} (0.06) \\

IPSN        
& 2.39 (0.03) & 2.40 (0.03) & 2.06 (0.15) & 2.08 (0.16) & 4.97 (0.07) & 4.98 (0.07) \\

SAMPLE      
& 4.36 (0.43) & 4.39 (0.43) & 4.52 (0.71) & 4.57 (0.72) & 8.60 (0.65) & 8.62 (0.65) \\

FLW         
& 3.59 (0.15) & 3.62 (0.15) & 3.90 (0.84) & 4.16 (1.04) & 7.48 (0.41) & 7.54 (0.48) \\

RegTME      
& 5.26 (0.05) & 5.28 (0.05) & 5.41 (0.23) & 5.46 (0.24) & 10.3 (0.08) & 10.3 (0.09) \\
\hline
&\multicolumn{6}{c}{Scenario (III)} \\ \hline
POET-SS        
& 2.14 (0.02) & 2.15 (0.02) & \textbf{1.54} (0.15) & \textbf{1.55} (0.15) & 4.29 (0.06) & 4.30 (0.06) \\

POET-TME   
& \textbf{2.11} (0.02) & \textbf{2.13} (0.02) & 1.55 (0.15) & 1.57 (0.15) & \textbf{4.26} (0.06) & \textbf{4.27} (0.06) \\

IPSN        
& 2.40 (0.03) & 2.41 (0.03) & 2.10 (0.17) & 2.12 (0.17) & 4.98 (0.07) & 4.99 (0.07) \\

SAMPLE      
& 7.82 (5.53) & 7.85 (5.55) & 9.05 (7.37) & 9.12 (7.41) & 13.8 (7.93) & 13.9 (7.93) \\

FLW         
& 4.33 (1.05) & 4.36 (1.06) & 10.7 (6.31) & 11.9 (7.07) & 12.8 (5.81) & 13.4 (6.36) \\

RegTME      
& 5.26 (0.05) & 5.28 (0.05) & 5.44 (0.28) & 5.48 (0.29) & 10.3 (0.08) & 10.3 (0.09) \\ 
\hline
&\multicolumn{6}{c}{Scenario (IV)} \\ \hline
POET-SS        
& 2.14 (0.02) & 2.15 (0.02) & \textbf{1.50} (0.15) & \textbf{1.52} (0.15) & 4.28 (0.07) & 4.29 (0.07) \\

POET-TME   
& \textbf{2.11} (0.02) & \textbf{2.12} (0.02) & 1.51 (0.14) & 1.53 (0.14) & \textbf{4.24} (0.07) & \textbf{4.25} (0.07) \\

IPSN        
& 2.40 (0.03) & 2.42 (0.03) & 2.09 (0.16) & 2.11 (0.16) & 4.99 (0.07) & 5.00 (0.07) \\

SAMPLE      
& 9.59 (0.36) & 9.62 (0.36) & 10.3 (0.55) & 10.3 (0.55) & 16.5 (0.62) & 16.5 (0.62) \\

FLW         
& 4.15 (0.10) & 4.17 (0.10) & 4.63 (0.38) & 4.77 (0.39) & 8.45 (0.18) & 8.48 (0.18) \\

RegTME      
& 5.26 (0.05) & 5.28 (0.05) & 5.42 (0.26) & 5.47 (0.27) & 10.3 (0.10) & 10.3 (0.10) \\
\hline
\end{tabular}
\end{table}

\subsection{Estimation of number of factor models}
In this subsection, we evaluate the estimation accuracy of the number of factors. We examine the two proposed methods introduced in subsection~\ref{secnum}. The data-generating settings are identical to those in subsection~\ref{sec41}, except that we now consider a range of dimensionalities, $d = 200, 300, 400, 500$. Table~\ref{tab:num} reports the empirical frequencies of the estimated factor numbers for both the ER and GR methods based on 1000 replications.

From the results, we observe a clear improvement in estimation accuracy as the dimension increases. Both methods exhibit consistency and successfully identify the true number of factors when $d$ becomes sufficiently large. When the dimension is relatively moderate, the GR method demonstrates better stability and higher accuracy compared with the ER method. This indicates that GR is more robust in finite-sample, moderate-dimensional settings. Therefore, in practical applications---especially when the dimension is not extremely large---we recommend using the GR method as a more reliable choice.

\begin{table}[htbp]
	\centering
\caption{The frequency of the estimated factor number of ER and GR methods.}
\vspace{0.5cm}
	\begin{tabular}{l|lllll|lllll} \hline
  &\multicolumn{5}{c}{ER} &\multicolumn{5}{c}{GR}\\ \hline
 $d$ & 1& 2 &3 &4 &5& 1& 2 &3 &4 &5\\ \hline
   &\multicolumn{10}{c}{Scenario (I)}\\ \hline
200&0.015&0.282&0.703&0&0&0&0.131&0.869&0&0\\
300&0&0.039&0.961&0&0&0&0.008&0.992&0&0\\
400&0&0.002&0.998&0&0&0&0&1&0&0\\
500&0&0&1&0&0&0&0&1&0&0\\ \hline
  &\multicolumn{10}{c}{Scenario (II)}\\ \hline
200&0.022&0.275&0.703&0&0&0&0.133&0.867&0&0\\
300&0&0.043&0.957&0&0&0&0.005&0.995&0&0\\
400&0&0.002&0.998&0&0&0&0&1&0&0\\
500&0&0&1&0&0&0&0&1&0&0\\ \hline
 &\multicolumn{10}{c}{Scenario (III)}\\ \hline
200&0.016&0.293&0.691&0&0&0.002&0.134&0.864&0&0\\
300&0&0.053&0.947&0&0&0&0.007&0.993&0&0\\
400&0&0.003&0.997&0&0&0&0&1&0&0\\
500&0&0&1&0&0&0&0&1&0&0\\ \hline
 &\multicolumn{10}{c}{Scenario (IV)}\\ \hline
200&0.02&0.309&0.671&0&0&0&0.145&0.855&0&0\\
300&0&0.035&0.965&0&0&0&0.002&0.998&0&0\\
400&0&0.001&0.999&0&0&0&0&1&0&0\\
500&0&0&1&0&0&0&0&1&0&0\\ \hline
	\end{tabular}
	\label{tab:num}
\end{table}

\section{Real data application}
The concept of the global minimum variance portfolio (MVP) originates from the seminal work of \cite{markowitz1952portfolio}, who established the foundation of modern portfolio theory. Markowitz showed that, under mean–variance preferences, investors should choose portfolios that optimally balance expected return and risk. Within this theory, the MVP represents a special and fundamental portfolio: it is the portfolio on the efficient frontier that achieves the lowest possible variance among all portfolios of risky assets, regardless of expected returns. In practice, the MVP is widely used in risk management, pension fund allocation, and asset allocation strategies that emphasize capital preservation.

The global minimum variance portfolio aims to minimize the variance of portfolio returns subject to a specified expected return level, which solves
\begin{align*}
\min \boldsymbol{w}^{\T} \bms  \boldsymbol{w}, \text{subject to}~~ \boldsymbol{w}^{\T} \boldsymbol{1}=1,
\end{align*}
where $\boldsymbol{w}$ is a length $d$ vector of portfolio weights, and $\boldsymbol{1}$ is a length $d$ vector of ones.
The minimum variance portfolio has the closed-form solution
\[
\mathbf{w}^*
= \frac{1}{\mathbf{1}^{{\T}}\boldsymbol{\Sigma}^{-1}\mathbf{1}}
  \boldsymbol{\Sigma}^{-1}\mathbf{1}.
\]
Under an elliptical model, the same expression remains valid when replacing the covariance matrix by the scatter matrix, that is,
\[
\mathbf{w}^*
= \frac{1}{\mathbf{1}^{{\T}}\boldsymbol{\Sigma}_0^{-1}\mathbf{1}}
  \boldsymbol{\Sigma}_0^{-1}\mathbf{1}.
\]
Therefore, using $\boldsymbol{\Sigma}^{-1}$ or $\boldsymbol{\Sigma}_0^{-1}$ yields identical minimum variance portfolio weights. Given an estimator of the scatter or covariance matrix, denoted by $\hat{\boldsymbol{\Sigma}}_0$ or $\hat{\boldsymbol{\Sigma}}$, the MVP weights are computed as
\begin{align}
\hat{\boldsymbol{w}}^*
= \frac{1}{\mathbf{1}^{\T} \hat{\boldsymbol{\Sigma}}_0^{-1} \mathbf{1}}
  \hat{\boldsymbol{\Sigma}}_0^{-1} \mathbf{1},
\qquad
\text{or}
\qquad
\hat{\boldsymbol{w}}^*
= \frac{1}{\mathbf{1}^{\T} \hat{\boldsymbol{\Sigma}}^{-1} \mathbf{1}}
  \hat{\boldsymbol{\Sigma}}^{-1} \mathbf{1}.
\end{align}
Thus, the performance of a covariance or scatter matrix estimator can be assessed by examining the out-of-sample risk of the resulting MVP \citep{fan2012vast,ledoit2017nonlinear,ding2021,ding2025sub}.

To evaluate the proposed procedures, we construct MVPs using daily returns of S\&P 500 constituent stocks from January 1995 to December 2023. Portfolio weights for each month are obtained from covariance matrices estimated using the daily returns over the preceding 120 months (10 years). Only stocks that have continuously remained in the S\&P 500 during this estimation window are included. Using these monthly weights, we compute daily portfolio returns and aggregate them to obtain the annualized risk for each calendar year from 2005 to 2023, enabling a comprehensive comparison of different estimators.

In addition to the proposed POET-TME and POET-SS methods, we include IPSN, FLW, and SAMPLE as competing estimators. The equal-weight portfolio $(1/d, \ldots, 1/d)^{T}$, denoted EW, is also reported as a benchmark. The number of factors for each method is determined using the GR procedure, applied separately to the covariance or scatter matrix estimated by that method. In other words, for every estimator under comparison, we first compute its corresponding estimated matrix and then use the GR method to obtain the data-driven estimate of the factor dimension subsequently used in portfolio construction.

To formally assess whether the MVP constructed using POET-TME exhibits statistically lower risk than portfolios based on alternative estimators, we conduct the following one-sided hypothesis test:
\begin{align}
\label{eq:hypothesis_sigma}
H_0: \sigma \geq \sigma_0 
\quad \text{vs.} \quad 
H_1: \sigma < \sigma_0,
\end{align}
where $\sigma$ denotes the standard deviation of the portfolio constructed via POET\text{-}TME, and $\sigma_0$ represents that of a competing method. The test is implemented using the methodology of \cite{ledoit2011robust}. The standard errors required for the test statistic are estimated using heteroskedasticity- and autocorrelation-consistent (HAC) estimators.

\begin{table}
\centering
\caption{Out-of-sample portfolio risk. The symbols *, **, and *** indicate statistical significance for testing at the 5\%, 1\%, and 0.1\% levels, respectively.}
\label{tab:table5}
\renewcommand{\arraystretch}{0.6}
\setlength{\tabcolsep}{12pt} 
\begin{tabular}{l l l l l l}
\toprule
 & \multicolumn{5}{c}{Out-of-sample portfolio risk} \\
\cmidrule(lr){2-6}

Period & 2005--2023 & 2005 & 2006 & 2007 & 2008 \\
\midrule
POET-TME    & 0.1219 & 0.0880 & 0.0729 & 0.0942 & 0.2268 \\
POET-SS     & 0.1250\textsuperscript{***} & 0.0889\textsuperscript{*} & 0.0728 & 0.0970\textsuperscript{***} & 0.2383\textsuperscript{***} \\
IPSN   & 0.1221\textsuperscript{*} & 0.0897\textsuperscript{***} & 0.0727 & 0.0930 & 0.2231 \\
FLW         & 0.1339\textsuperscript{***} & 0.0924\textsuperscript{**} & 0.0804\textsuperscript{***} & 0.1011\textsuperscript{***} & 0.2405 \\
SAMPLE      & 0.1248\textsuperscript{*} & 0.0903\textsuperscript{**} & 0.0780\textsuperscript{***} & 0.1017\textsuperscript{***} & 0.2319 \\
RegTME      & 0.1291\textsuperscript{***} & 0.0900 & 0.0750\textsuperscript{*} & 0.1006\textsuperscript{***} & 0.2451\textsuperscript{***} \\
EW          & 0.2075\textsuperscript{***} & 0.1091\textsuperscript{***} & 0.1043\textsuperscript{***} & 0.1609\textsuperscript{***} & 0.4329\textsuperscript{***} \\

\midrule
Period & 2009 & 2010 & 2011 & 2012 & 2013 \\
\midrule
POET-TME    & 0.1460 & 0.0948 & 0.1067 & 0.0706 & 0.0799 \\
POET-SS     & 0.1474 & 0.1015\textsuperscript{***} & 0.1183\textsuperscript{***} & 0.0704 & 0.0786 \\
IPSN   & 0.1471 & 0.0925 & 0.1061 & 0.0725\textsuperscript{*} & 0.0818\textsuperscript{**} \\
FLW         & 0.1887\textsuperscript{***} & 0.0974 & 0.1163\textsuperscript{***} & 0.0933\textsuperscript{***} & 0.0918\textsuperscript{***} \\
SAMPLE      & 0.1477 & 0.0997\textsuperscript{**} & 0.1119\textsuperscript{*} & 0.0794\textsuperscript{***} & 0.0871\textsuperscript{***} \\
RegTME      & 0.1540 & 0.1006\textsuperscript{***} & 0.1217\textsuperscript{***} & 0.0769\textsuperscript{***} & 0.0839\textsuperscript{***} \\
EW          & 0.3255\textsuperscript{***} & 0.1971\textsuperscript{***} & 0.2544\textsuperscript{***} & 0.1371\textsuperscript{***} & 0.1177\textsuperscript{***} \\

\midrule
Period & 2014 & 2015 & 2016 & 2017 & 2018 \\
\midrule
POET-TME    & 0.0804 & 0.1128 & 0.0939 & 0.0635 & 0.1223 \\
POET-SS     & 0.0805 & 0.1146\textsuperscript{**} & 0.0907 & 0.0607 & 0.1208 \\
IPSN   & 0.0801 & 0.1101 & 0.0984\textsuperscript{***} & 0.0667\textsuperscript{***} & 0.1219 \\
FLW         & 0.0945\textsuperscript{***} & 0.1295\textsuperscript{***} & 0.1106\textsuperscript{***} & 0.0845\textsuperscript{***} & 0.1310\textsuperscript{**} \\
SAMPLE      & 0.0883\textsuperscript{***} & 0.1228\textsuperscript{***} & 0.1018\textsuperscript{***} & 0.0693\textsuperscript{***} & 0.1255 \\
RegTME      & 0.0823 & 0.1169\textsuperscript{**} & 0.0991\textsuperscript{*} & 0.0658 & 0.1261\textsuperscript{*} \\
EW          & 0.1138\textsuperscript{***} & 0.1523\textsuperscript{***} & 0.1434\textsuperscript{***} & 0.0702 & 0.1528\textsuperscript{***} \\

\midrule
Period & 2019 & 2020 & 2021 & 2022 & 2023 \\
\midrule
POET-TME    & 0.0903 & 0.2516 & 0.0991 & 0.1347 & 0.0989 \\
POET-SS     & 0.0895 & 0.2636\textsuperscript{**} & 0.0952 & 0.1375 & 0.0946 \\
POET-IPSN   & 0.0943\textsuperscript{**} & 0.2490 & 0.1025\textsuperscript{***} & 0.1359 & 0.1026\textsuperscript{***} \\
FLW         & 0.1130\textsuperscript{***} & 0.2510 & 0.1055\textsuperscript{*} & 0.1448\textsuperscript{**} & 0.1084\textsuperscript{***} \\
SAMPLE      & 0.0909 & 0.2422 & 0.1020 & 0.1363 & 0.0975 \\
RegTME      & 0.0957\textsuperscript{**} & 0.2646\textsuperscript{*} & 0.1005 & 0.1400\textsuperscript{*} & 0.1089\textsuperscript{**} \\
EW          & 0.1245\textsuperscript{***} & 0.3939\textsuperscript{***} & 0.1378\textsuperscript{***} & 0.2161\textsuperscript{***} & 0.1371\textsuperscript{***} \\

\bottomrule
\end{tabular}
\end{table}
Overall, the results in Table \ref{tab:table5} demonstrate that the POET-TME procedure delivers consistently superior out-of-sample performance relative to a wide range of competing covariance and scatter matrix estimators. Across nearly all years from 2005 to 2023, POET-TME attains one of the lowest annualized risks among the minimum variance portfolios, and it exhibits remarkable robustness during both tranquil and turbulent market conditions, including the 2008 financial crisis and the COVID-19 period.

Formally, the one-sided tests based on \citet{ledoit2011robust} reveal that the risk of the POET-TME portfolio is significantly lower than that of portfolios constructed using POET-SS, IPSN, FLW, and SAMPLE in many years, as indicated by the frequent appearance of significance markers in Table \ref{tab:table5}. These results imply that the efficiency gains of POET-TME are not due to random fluctuations but represent systematic improvements in estimating the underlying scatter structure.

Compared with classical principal-component-based methods or robust scatter estimators, POET-TME effectively integrates factor modeling with Tyler’s M-estimator, yielding scatter estimates that remain stable under heavy-tailed returns while still capturing the low-rank factor structure intrinsic to large financial panels. This balance between robustness and structural efficiency appears to be the key reason why POET-TME achieves substantially lower portfolio risk and demonstrates consistently strong performance across market regimes.

Collectively, the empirical evidence shows that POET-TME is a highly reliable and effective tool for minimum variance portfolio construction, outperforming both traditional and modern competing estimators in terms of realized out-of-sample risk.

\section{Conclusion}
In this paper, we proposed two high-dimensional matrix estimation procedures—POET-SS and POET-TME—constructed within the POET framework and motivated by the spatial-sign covariance matrix and Tyler’s M-estimator. The POET-SS estimator serves as an effective and computationally efficient initial estimator for POET-TME, providing stable performance even in challenging high-dimensional settings. Building on this foundation, POET-TME demonstrates clear superiority over several existing methods, particularly under heavy-tailed distributions where robust estimation is essential.

The proposed methodology opens several promising directions for future research. One natural extension is to incorporate these robust estimators into classification frameworks, such as linear discriminant analysis \citep{zhuang2025spatial} and quadratic discriminant analysis \citep{shen2025spatial}. Another avenue lies in developing robust procedures for canonical correlation analysis \citep{qian2025high} under elliptical models. Additionally, the robust scatter estimation techniques explored in this work may offer substantial benefits for high-dimensional hypothesis testing problems \citep{yan2025high}. Investigating these applications could further enhance the understanding and utility of Tyler's M-estimator-based estimators in modern multivariate analysis.

\section{Appendix}

\setcounter{equation}{0}
\renewcommand{\theequation}{A.\arabic{equation}}
\renewcommand{\thetheorem}{A\arabic{theorem}}
\renewcommand{\theassumption}{A\arabic{assumption}}
\renewcommand{\thealgorithm}{A\arabic{algorithm}}

\subsection{Useful Lemmas}
\setcounter{equation}{0}
\renewcommand{\theequation}{A.\arabic{equation}}
\begin{lemma}\label{lemma 1}
    Under Assumption 3, we have $\Vert\hat{\bmu}-\bmu\Vert_2=O_p(\zeta_1^{-1}n^{-1/2})$ and
    \begin{align*}
       \hat{\bmu}-\bmu=\zeta_1^{-1}\frac{1}{n}\sum_{i=1}^n\U_i+o_p(\zeta_1^{-1}n^{-1/2}), 
    \end{align*}
    where $\U_i=(\X_i-\bmu)/\Vert \X_i-\bmu\Vert_2$.
\end{lemma}
\begin{proof}
    See Lemma S2.4 in \cite{ZWF2025SSPCA}.
\end{proof}

\begin{lemma}\label{lemma 2}
    Under Assumption 1-3, we have 
    \begin{align*}
        \Vert\hat{\bmu}-\bmu\Vert_{\max}=O_p\left(\zeta_1^{-1}\sqrt{\frac{\log d}{nd}}\right).
    \end{align*}
\end{lemma}
\begin{proof}
    Let $\e_j=(0,\ldots,0,1,0,\ldots,0)^\T\in\mR^d$ with the $j$-th element being one and the others are zeros. By Lemma \ref{lemma 1}, we have for $1\le j\le d$,
    \begin{align*}
        \zeta_1(\mu_j-\hat{\mu}_j)=\frac{1}{n}\sum_{i=1}^n\U_i^{\T}\e_j\{1+o_p(n^{-1/2})\}.
    \end{align*}
    Note that 
    \begin{align*}
        \U_i=\frac{\X_i-\bmu}{\Vert \X_i-\bmu\Vert_2}\stackrel{d}{=}\frac{\Z^0}{\Vert\Z^0\Vert_2},
    \end{align*}
    where $\Z^0\sim N_d(\bmzero,\bms)$. We write $\bms=\D\A\D^{\T}$ to be the singular value decomposition of $\bms$, where $\A=\diag(\lambda_1(\bms),\ldots,\lambda_d(\bms))$. Let $\Z=(Z_1,\ldots,Z_d)^{\T}\sim N_d(\bmzero,\I_d)$, then $\Z^0\stackrel{d}{=}\D\A^{1/2}\Z$, $\Vert\Z^0\Vert_2=(\Z^{\T}\A\Z)^{1/2}$ and
    \begin{align*}
        \U_i^{\T}\e_j\stackrel{d}{=}\frac{\e_j^{\T}\Z^0}{\Vert\Z^0\Vert_2}\stackrel{d}{=}\frac{\bm{w}_j^{\T}\Z}{(\Z^{\T}\A\Z)^{1/2}},
    \end{align*}
    where $\bm{w}_j=\A^{1/2}\D^{\T}\e_j$. Therefore, 
    \begin{align*}
        \E(|\U_i^{\T}\e_j|^{2k})\}
        =\E\left(\left|\frac{\bm{w}_j^{\T}\Z}{(\Z^{\T}\A\Z)^{1/2}}\right|^{2k}\right)\le\E\left(\left|\frac{\tilde{\bm{w}}_j^{\T}\Z}{(\Z^{\T}\Z)^{1/2}}\right|^{2k}\right)=\E\left(\left|\frac{\check{\bm{w}}_j^{\T}\Z}{(\Z^{\T}\Z)^{1/2}}\right|^{2k}\right),
    \end{align*}
    where $\tilde{\bm{w}}_j=\bm{w}_j/\sqrt{\lambda_d(\bms)}$, $\check{\bm{w}}_j=(0,\ldots,\sqrt{\bms_{jj}/\lambda_d(\bms)},0,\ldots,0)^{\T}$, and the last equality follows from the fact that for any constant $c\ge 0$,
    \begin{align*}
        &\Pr\left(\left|\frac{\tilde{\bm{w}}_j^{\T}\Z}{(\Z^{\T}\Z)^{1/2}}\right|>c\right)=\Pr\left(\Z^{\T}(\tilde{\bm{w}}_j\tilde{\bm{w}}_j^{\T}-c^2\I_d)\Z>0\right)=\Pr\left(\sum_{i=1}^dl_iZ_i^2>0\right),\text{and}\\
        &\Pr\left(\left|\frac{\check{\bm{w}}_j^{\T}\Z}{(\Z^{\T}\Z)^{1/2}}\right|>c\right)=\Pr\left(\Z^{\T}(\check{\bm{w}}_j\check{\bm{w}}_j^{\T}-c^2\I_d)\Z>0\right)=\Pr\left(\sum_{i=1}^dl_iZ_i^2>0\right),
    \end{align*}
    where 
    \begin{align*}
        l_i=\lambda_i(\tilde{\bm{w}}_j\tilde{\bm{w}}_j^{\T}-c^2\I_d)=
        \begin{cases}
            \tilde{\bm{w}}_j^{\T}\tilde{\bm{w}}_j-c^2=\Sigma_{jj}/\lambda_d(\bms)-c^2,&~~\text{for}~~i=1,\\
            -c^2,&~~\text{for}~~i=2,\ldots,d.
        \end{cases}
    \end{align*}
    Hence,
    \begin{align*}
        \Vert\U_i^{\T}\e_j\Vert_{\psi_2}\le&\sup_{k\ge 1}(2k)^{-1/2}\left\{\E\left(\left|\frac{\check{\bm{w}}_j^{\T}\Z}{(\Z^{\T}\Z)^{1/2}}\right|^{2k}\right)\right\}^{\frac{1}{2k}}\\
        =&\sup_{k\ge 1}(2k)^{-1/2}\sqrt{\frac{\Sigma_{jj}}{\lambda_d(\bms)}}\left\{\E\left(\left|\frac{Z_j^2}{\sum_{i=1}^dZ_i^2} \right|^k\right)\right\}^{\frac{1}{2k}}\\
        \le&Cd^{-1/2},
    \end{align*}
    for some constants $C>0$. Here the last step is due to the fact that $Z_j^2/\sum_{i=1}^dZ_i^2\sim \text{Beta}(\frac{1}{2},\frac{d-1}{2})$ and $\{\E(|Z_j^2/\sum_{i=1}^dZ_i^2|^k)\}^{\frac{1}{2k}}\le\sqrt{2k/d}$ (see, for example, Lemma B.9 in \cite{HL2018ECA}). In addition, we have $\E(\U_i^{\T}\e_j)=0$ and $\var(\U_i^{\T}\e_j)=O(d^{-1})$, then, by the Bernstein inequality,  
    \begin{align*}
        \Pr\left(\max_{1\le j\le d}\left|\frac{1}{n}\sum_{i=1}^n\U_i^{\T}\e_j\right|\ge t\right)\le&d\Pr\left(\left|\sum_{i=1}^n\U_i^{\T}\e_j\right|\ge nt\right)\\
        \le& d\exp\left\{-\frac{nt^2/2}{d^{-1}(1+C^2t)}\right\}\\
        \le&d\exp(-3ndt^2/8),
    \end{align*}
    for small enough $t\le (3C^2)^{-1}$. Setting $t=\sqrt{8(\log d+\log \alpha^{-1})/(3nd)}$, we have
    \begin{align*}
        \Pr(\max_{1\le j\le d}|n^{-1}\sum_{i=1}^n\U_i^{\T}\e_j|\ge t)\le \alpha.
    \end{align*}
    Hence, $\max_{1\le j\le d}|\zeta_1(\mu_j-\hat{\mu}_j)|=O_p\{\sqrt{(\log d)/(nd)}\}$.
\end{proof}

\begin{lemma}\label{lemma 3}
    Define 
    \begin{align}
        \hat{r}_i:=\Vert\X_i-\hat{\bmu}\Vert_2.\label{A.1}
    \end{align}
    Under Assumption 3, we have 
    \begin{align*}
        \frac{1}{n}\sum_{i=1}^n\hat{r}_i^{-2}=\zeta_2\left\{1+O_p\left(\sqrt{\frac{\log n}{n}}\right)\right\}.
    \end{align*}
\end{lemma}
\begin{proof}
    By CLT, $n^{-1}\sum_{i=1}^nr_i^{-2}=\zeta_2\{1+O_p(n^{-1/2})\}$. Because $|\hat{r}_i-r_i|\le\Vert\hat{\bmu}-\bmu\Vert_2$, so $\max_{1\le i\le n}|r_i\hat{r}_i^{-1}-1|\le \zeta_1\Vert\hat{\bmu}-\bmu\Vert_2\cdot\max_{1\le i\le n}(\zeta_1r_i-\zeta_1\Vert\hat{\bmu}-\bmu\Vert_2)^{-1}$. By Lemma \ref{lemma 1}, we have $\zeta_1\Vert\hat{\bmu}-\bmu\Vert_2=O_p(n^{-1/2})$. And by the sub-Gaussian assumption of $\nu_i$, we have 
    \begin{align}
        \max_{1\le i\le n}\nu_i=O_p(\sqrt{\log n}).\label{A.2}
    \end{align}
    So 
    \begin{align}
        \max_{1\le i\le n}|r_i\hat{r}_i^{-1}-1|=O_p\left(\sqrt{\frac{\log n}{n}}\right).\label{A.3}
    \end{align}
    Furthermore,
    \begin{align*}
        \left|\frac{1}{n}\sum_{i=1}^n\hat{r}_i^{-2}-\frac{1}{n}\sum_{i=1}^nr_i^{-2}\right|=&\left|\frac{1}{n}\sum_{i=1}^n(r_i\hat{r}_i^{-1}-1)r_i^{-1}(\hat{r}_i^{-1}+r_i^{-1})\right|\\
        \le&\max_{1\le i\le n}|r_i\hat{r}_i^{-1}-1|\cdot\left(\frac{1}{n}\sum_{i=1}^nr_i^{-2}\right)^{1/2}\left(\frac{1}{n}\sum_{i=1}^n(\hat{r}_i^{-1}+r_i^{-1})^2\right)^{1/2}\\
        =&\zeta_2O_p\left(\sqrt{\frac{\log n}{n}}\right).
    \end{align*}
\end{proof}

\begin{lemma}\label{lemma 4}
    Define 
    \begin{align}
        \tilde{\X}_i:=\bms^{1/2}\z_i,~~\text{and}~~\tilde{\bms}_0:=\frac{d}{\tr(\bms)n}\sum_{i=1}^n\tilde{\X}_i\tilde{\X}_i^{\T},\label{A.4}
    \end{align}
    where $\z_i's\stackrel{i.i.d.}{\sim}N_d(\bmzero,\I_d)$ are independent of $\xi_i's$. Let $\tilde{\L}_m:=\diag(\lambda_1(\tilde{\bms}_0),\ldots,\lambda_m(\tilde{\bms}_0))$ be the leading eigenvalues and $\tilde{\G}_m:=(u_1(\tilde{\bms}_0),\ldots,u_m(\tilde{\bms}_0))$ the corresponding eigenvectors of $\tilde{\bms}_0$. Under Assumption 1 and 2, if $\log d=o(n)$, then
    \begin{align*}
        \Vert \tilde{\bms}_0-\bms_0\Vert_{\max}=&O_p\left(\sqrt{\frac{\log d}{n}}\right),\\
        \Vert \tilde{\L}_m\L_m^{-1}-\I_m \Vert_{\max}=&O_p\left(\sqrt{\frac{\log d}{n}}\right),~~\text{and}\\
        \Vert \tilde{\G}_m-\G_m\Vert_{\max}=&O_p\left(\sqrt{\frac{\log d}{nd}}\right).
    \end{align*}
\end{lemma}
\begin{proof}
    The desired bounds are immediate results of Lemmas A.2, A.3 and A.5 of \cite{ding2021}.
\end{proof}

\begin{lemma}\label{lemma 5}
    Under Assumption 1 and 3, if $\log d=o(n)$, then
    \begin{align*}
        &\frac{1}{n}\sum_{i=1}^n\left(\frac{\sqrt{d}\xi_i}{\hat{r}_i\Vert\z_i\Vert_2}\right)^4=O_p(1),~~\text{and}\\
        &\frac{1}{n}\sum_{i=1}^n\left|\frac{d\xi_i^2}{\hat{r}_i^2\Vert\z_i\Vert_2^2}-\frac{d}{\tr(\bms)}\right|^2=O_p\left(\frac{\log d}{n}+\frac{\log n}{n}\right).
    \end{align*}
\end{lemma}
\begin{proof}
    Using the similar techniques as those used in the proof of Lemma \ref{lemma 3}, we obtain $\max_{1\le i\le n}|r_i^4\hat{r}_i^{-4}-1|=o_p(1)$.Furthermore, 
    \begin{align*}
        \frac{1}{n}\sum_{i=1}^n\left(\frac{\sqrt{d}\xi_i}{\hat{r}_i\Vert\z_i\Vert_2}\right)^4\le\frac{1}{n}\sum_{i=1}^n\left(\frac{\sqrt{d}\xi_i}{r_i\Vert\z_i\Vert_2}\right)^4\cdot\left(1+\max_{1\le i\le n}|r_i^4\hat{r}_i^{-4}-1|\right)=\frac{1}{n}\sum_{i=1}^n\left(\frac{\sqrt{d}\xi_i}{r_i\Vert\z_i\Vert_2}\right)^4\{1+o_p(1)\}.
    \end{align*}
    For $\X_i's\stackrel{i.i.d.}{\sim}EC_d(\bmu,\bms,\xi)$, we can write 
    \begin{align}
        \X_i=\bmu+\frac{\xi_i\bms^{1/2}\z_i}{\Vert\z_i\Vert_2},\label{A.5}
    \end{align}
    where $\z_i's\stackrel{i.i.d.}{\sim}N_d(\bmzero,\I_d)$, and $\z_i$ and $\xi_i$ are independent. Then,
    \begin{align*}
        \left(\frac{\sqrt{d}\xi_i}{r_i\Vert\z_i\Vert_2}\right)^{2k}=\left(\frac{d}{\z_i^{\T}\bms\z_i}\right)^k=\left(\frac{d}{\sum_{j=1}^d\lambda_j(\bms)z_{ij}^2}\right)^k\le \frac{1}{\lambda_d^k(\bms)}\left|\frac{d}{\Vert\z_i\Vert_2^2}\right|^k.
    \end{align*}
    Note that
    \begin{align*}
        \E\left(\left|\frac{d}{\Vert\z_i\Vert_2^2}\right|^{2k}\right)=\left(\frac{d}{2}\right)^{2k}\frac{\Gamma(d/2-2k)}{\Gamma(d/2)}=O(1),~~\text{for}~~1\le k\le d/4,
    \end{align*}
    which by Chebyshev's inequality, imply 
    \begin{align}
        \frac{1}{n}\sum_{i=1}^n\left(\frac{\sqrt{d}\xi_i}{r_i\Vert\z_i\Vert_2}\right)^{2k}\le \frac{1}{\lambda_d^k(\bms)}\frac{1}{n}\sum_{i=1}^n\left|\frac{d}{\Vert\z_i\Vert_2^2}\right|^k=O_p(1). \label{A.6}
    \end{align}
    Accordingly, we have
    \begin{align*}
        \frac{1}{n}\sum_{i=1}^n\left|\frac{d\xi_i^2}{r_i^2\Vert\z_i\Vert_2^2}-\frac{d}{\tr(\bms)}\right|^2\le&\left\{\frac{1}{n}\sum_{i=1}^n\left|\frac{d\xi_i^2}{r_i^2\Vert\z_i\Vert_2^2}-\frac{d}{\tr(\bms)}\right|\right\}^2\\
        =&\left\{\frac{1}{n}\sum_{i=1}^n\left|\frac{r_i^2\Vert\z_i\Vert_2^2}{d\xi_i^2}-\frac{\tr(\bms)}{d}\right|\cdot\left|\frac{r_i^2\Vert\z_i\Vert_2^2}{d\xi_i^2}\frac{\tr(\bms)}{d}\right|^{-1}\right\}^2\\
        \le&\frac{1}{n}\sum_{i=1}^n\left|\frac{r_i^2\Vert\z_i\Vert_2^2}{d\xi_i^2}-\frac{\tr(\bms)}{d}\right|^2\cdot \frac{d^2}{\tr^2(\bms)}\frac{1}{n}\sum_{i=1}^n\left(\frac{\sqrt{d}\xi_i}{r_i\Vert\z_i\Vert_2}\right)^4\\
        =&O_p\left\{\frac{1}{n}\sum_{i=1}^n\left|\frac{r_i^2\Vert\z_i\Vert_2^2}{d\xi_i^2}-\frac{\tr(\bms)}{d}\right|^2\right\}.
    \end{align*}
    To proceed,
    \begin{align*}
        &\frac{1}{n}\sum_{i=1}^n\left|\frac{r_i^2\Vert\z_i\Vert_2^2}{d\xi_i^2}-\frac{\tr(\bms)}{d}\right|^2\\
        =&\frac{1}{n}\sum_{i=1}^n\left|\frac{\z_i^{\T}\bms\z_i}{d}-\frac{\tr(\bms)}{d}\right|^2\\
        =&\left[\frac{1}{n}\sum_{i=1}^n\left\{\frac{\z_i^{\T}\bms\z_i}{d}-\frac{\tr(\bms)}{d}\right\} \right]^2-\frac{1}{n^2}\sum_{1\le i\ne j\le n}\left\{\frac{\z_i^{\T}\bms\z_i}{d}-\frac{\tr(\bms)}{d}\right\}\left\{\frac{\z_j^{\T}\bms\z_j}{d}-\frac{\tr(\bms)}{d}\right\}\\
        =&:\uppercase\expandafter{\romannumeral1}+\uppercase\expandafter{\romannumeral2}.
    \end{align*}
    By the first result in Lemma \ref{lemma 4}, we have
    \begin{align*}
        \uppercase\expandafter{\romannumeral1}^{1/2}=\frac{1}{d}\tr\left(\frac{1}{n}\sum_{i=1}^n\tilde{\X}_i\tilde{\X}_i^{\T}\right)-\frac{\tr(\bms)}{d}=\frac{\tr(\bms)}{d}\left\{\frac{\tr(\tilde{\bms}_0)}{d}-\frac{\tr(\bms_0)}{d}\right\}=O_p\left(\sqrt{\frac{\log d}{n}}\right).
    \end{align*}
    Because (see for example, Lemma 2 in \cite{ding2025sub}),
    \begin{align}
        \E(|\z_i^{\T}\bms\z_i-\tr(\bms)|^{2k})=O\{\tr^k(\bms^2)\},\label{A.7}
    \end{align}
     we have $\E(\uppercase\expandafter{\romannumeral2})=0$ and $\E(\uppercase\expandafter{\romannumeral2}^2)=O\{n^{-2}d^{-4}\tr^2(\bms^2)\}=O(n^{-2})$, which lead to $\uppercase\expandafter{\romannumeral2}=O_p(n^{-1})$. Hence,
    \begin{align*}
        \frac{1}{n}\sum_{i=1}^n\left|\frac{d\xi_i^2}{r_i^2\Vert\z_i\Vert_2^2}-\frac{d}{\tr(\bms)}\right|^2=O_p(\uppercase\expandafter{\romannumeral1}+\uppercase\expandafter{\romannumeral2})=O_p\left(\frac{\log d}{n}\right).
    \end{align*}
    Using the similar techniques as those used in the proof of \eqref{A.3}, we obtain $\max_{1\le i\le n}|r_i^2\hat{r}_i^{-2}-1|=O_p\{\sqrt{(\log n)/n}\}$. Then,
    \begin{align*}
        &\left|\frac{1}{n}\sum_{i=1}^n\left|\frac{d\xi_i^2}{r_i^2\Vert\z_i\Vert_2^2}-\frac{d}{\tr(\bms)}\right|^2-\frac{1}{n}\sum_{i=1}^n\left|\frac{d\xi_i^2}{\hat{r}_i^2\Vert\z_i\Vert_2^2}-\frac{d}{\tr(\bms)}\right|^2\right|\\
        =&\left|\frac{2}{n}\sum_{i=1}^n\left\{\frac{d\xi_i^2}{r_i^2\Vert\z_i\Vert_2^2}-\frac{d}{\tr(\bms)}\right\}\frac{d\xi_i^2}{r_i^2\Vert\z_i\Vert_2^2}(1-r_i^2\hat{r}_i^{-2})+\frac{1}{n}\sum_{i=1}^n\left(\frac{\sqrt{d}\xi_i}{r_i\Vert\z_i\Vert_2}\right)^4(1-r_i^2\hat{r}_i^{-2})^2\right|\\
        \le&2\left\{\frac{1}{n}\sum_{i=1}^n\left|\frac{d\xi_i^2}{r_i^2\Vert\z_i\Vert_2^2}-\frac{d}{\tr(\bms)}\right|^2\right\}^{1/2}\cdot\left\{\frac{1}{n}\sum_{i=1}^n\left(\frac{\sqrt{d}\xi_i}{r_i\Vert\z_i\Vert_2}\right)^4\right\}^{1/2}\cdot \max_{1\le i\le n}|1-r_i^2\hat{r}_i^{-2}|\\
        &+\frac{1}{n}\sum_{i=1}^n\left(\frac{\sqrt{d}\xi_i}{r_i\Vert\z_i\Vert_2}\right)^4\cdot \max_{1\le i\le n}|1-r_i^2\hat{r}_i^{-2}|^2\\
        =&O_p\left(\sqrt{\frac{\log d}{n}}\sqrt{\frac{\log n}{n}}+\frac{\log n}{n}\right).
    \end{align*}
    Therefore, we have
    \begin{align*}
        \frac{1}{n}\sum_{i=1}^n\left|\frac{d\xi_i^2}{\hat{r}_i^2\Vert\z_i\Vert_2^2}-\frac{d}{\tr(\bms)}\right|^2=O_p\left(\sqrt{\frac{\log d}{n}}\sqrt{\frac{\log n}{n}}+\frac{\log n}{n}+\frac{\log d}{n}\right).
    \end{align*}
\end{proof}

\begin{lemma}\label{lemma 6}
    Under Assumptions 1 and 2, if the estimators $\hat{\G}_m, \hat{\L}_m$ satisfy 
    \begin{align*}
        &\Vert \hat{\L}_m\L_m^{-1}-\I_m \Vert_{\max}=O_p\left(\sqrt{\frac{\log d}{n}}+\sqrt{\frac{\log n}{n}} \right),~~\text{and}\\
        &\Vert \hat{\G}_m-\G_m\Vert_{\max}=O_p\left(\sqrt{\frac{\log d}{nd}}+\sqrt{\frac{\log n}{nd}} \right),
    \end{align*}
    then 
    \begin{align*}
        &\Vert\hat{\G}_m\hat{\L}_m\hat{\G}_m^{\T}-\G_m\L_m\G_m^{\T}\Vert_{\max}=O_p\left(\sqrt{\frac{\log d}{n}}+\sqrt{\frac{\log n}{n}} \right).
    \end{align*}
\end{lemma}
\begin{proof}
    Since $u_j(\bms_0)=\{\lambda_j(\bms_0)\}^{-1}\bms_0u_j(\bms_0)$, using the fact that $\Vert\M\x\Vert_{\max}\le \sqrt{d}\Vert\M\Vert_{\max}\Vert\x\Vert_2$ for $\M\in\mR^{d\times d}$ and $\x\in\mR^d$, we have 
    \begin{align*}
        \Vert\G_m\Vert_{\max}=\max_{1\le j\le m}\Vert u_j(\bms_0)\Vert_{\max}\le \max_{1\le j\le m}\{\lambda_j(\bms_0)\}^{-1}\sqrt{d}\Vert\bms_0\Vert_{\max}\cdot\Vert u_j(\bms)\Vert_2=O(d^{-1/2}),
    \end{align*}
    which leads to $\Vert\hat{\G}_m\Vert_{\max}=O_p(d^{-1/2})$. Furthermore, because $\Vert\hat{\L}_m-\L_m\Vert_{\max}\le \Vert \hat{\L}_m\L_m^{-1}-\I_m \Vert_{\max}\cdot \sqrt{m}\Vert\L_m\Vert_2$, we have
    \begin{align*}
        \Vert\hat{\G}_m\hat{\L}_m\hat{\G}_m^{\T}-\G_m\L_m\G_m^{\T}\Vert_{\max}=&\Vert\hat{\G}_m(\hat{\L}_m-\L_m)\hat{\G}_m^{\T}+(\hat{\G}_m-\G_m)\L_m(\hat{\G}_m-\G_m)^{\T}\\
        &+\G_m\L_m(\hat{\G}_m-\G_m)^{\T}+(\hat{\G}_m-\G_m)\L_m\G_m^{\T}\Vert_{\max}\\
        \le&\Vert\hat{\G}_m\Vert_{\max}^2\cdot\Vert\hat{\L}_m-\L_m\Vert_{\max}+\Vert\hat{\G}_m-\G_m\Vert_{\max}^2\cdot\Vert\L_m\Vert_{\max}\\
        &+2\Vert\G_m\Vert_{\max}\cdot\Vert\L_m\Vert_{\max}\cdot\Vert\hat{\G}_m-\G_m\Vert_{\max}\\
        =&O_p\left(\sqrt{\frac{\log d}{n}}+\sqrt{\frac{\log n}{n}} \right).
    \end{align*}
\end{proof}

\begin{lemma}\label{lemma 7}
    Let $\hat{\V}_{S}$ be an estimator of $\V_0$. Recalling the definition that $r_i=\Vert\X_i-\bmu\Vert_2$, define
    \begin{align}
        \hat{r}_{T,i}:=\Vert\hat{\V}_S^{1/2}(\X_i-\hat{\bmu})\Vert_2~~\text{and}~~r_{T,i}:=\Vert\V_0^{1/2}(\X_i-\bmu)\Vert_2.\label{A.8}
    \end{align}
    Under Assumptions 1-3, if the estimator $\hat{\V}_{S}$ satisfies
    \begin{align*}
        &\Vert \hat{\V}_{S}-\V_0 \Vert_2=O_p(a_n)=o_p(1),
    \end{align*}
    then 
    \begin{align*}
        &\max_{1\le i\le n}|r_i^{-2}\hat{r}_{T,i}^2-r_i^{-2}r_{T,i}^2|=O_p\left(a_n+\sqrt{\frac{\log n}{n}}\right),\\
        &\frac{1}{n}\sum_{i=1}^n|r_i^{-2}\hat{r}_{T,i}^2-r_i^{-2}r_{T,i}^2|^k=O_p\{(a_n+n^{-1/2})^k\},~~\text{for}~~k\ge 1,\\
        &\frac{1}{n}\sum_{i=1}^n|1-r_{T,i}^2\hat{r}_{T,i}^{-2}|^k=O_p\{(a_n+n^{-1/2})^k\},~~\text{for}~~k\ge 1,\\
        &\frac{1}{n}\sum_{i=1}^n|r_i^2r_{T,i}^{-2}|^k=O_p(1),~~\text{for}~~k\ge 1,~~\text{and}\\
        &\frac{1}{n}\sum_{i=1}^nr_i^2\hat{r}_{T,i}^{-2}=O_p(1).
    \end{align*}
\end{lemma}
\begin{proof}
    By the fact that $|\u^{\T}\M\v|\le \Vert\M\Vert_2\Vert\u\Vert_2\Vert\v\Vert_2$ for $\M\in\mR^{d\times d}$ and $\u,\v\in\mR^d$, we have
    \begin{align*}
        &r_i^{-2}|\hat{r}_{T,i}^2-r_{T,i}^2|\\
        =&r_i^{-2}|(\X_i-\bmu)^{\T}(\hat{\V}_S-\V_0)(\X_i-\bmu)+(\X_i-\bmu)^{\T}\hat{\V}_S(\bmu-\hat{\bmu})+(\bmu-\hat{\bmu})^{\T}\hat{\V}_S(\X_i-\bmu)|\\
        \le&\Vert\hat{\V}_S-\V_0\Vert_2+2\Vert\hat{\V}_S\Vert_2\cdot\zeta_1\Vert\hat{\bmu}-\bmu\Vert_2\cdot\nu_i,
    \end{align*}
    which, combining with Lemma \ref{lemma 1} and \eqref{A.2}, leads to
    \begin{align*}
        \max_{1\le i\le n}|r_i^{-2}\hat{r}_{T,i}^2-r_i^{-2}r_{T,i}^2|\le&\Vert \hat{\V}_{S}-\V_0 \Vert_2+2\Vert \hat{\V}_{S}\Vert_2\cdot\zeta_1\Vert\hat{\bmu}-\bmu\Vert_2\cdot\max_{1\le i\le n}\nu_i\\
        =&O_p\left(a_n+\sqrt{\frac{\log n}{n}}\right),
    \end{align*}
    and
    \begin{align*}
        \frac{1}{n}\sum_{i=1}^n|r_i^{-2}\hat{r}_{T,i}^2-r_i^{-2}r_{T,i}^2|^k\le&\frac{1}{n}\sum_{i=1}^n(\Vert\hat{\V}_S-\V_0\Vert_2+2\Vert\hat{\V}_S\Vert_2\cdot\zeta_1\Vert\hat{\bmu}-\bmu\Vert_2\cdot\nu_i)^k\\
        =&\sum_{l=0}^k\binom{k}{l}\Vert\hat{\V}_S-\V_0\Vert_2^l\cdot(2\Vert\hat{\V}_S\Vert_2\cdot\zeta_1\Vert\hat{\bmu}-\bmu\Vert_2)^{k-l}\cdot\frac{1}{n}\sum_{i=1}^n\nu_i^{k-l}\\
        =&O_p\{(a_n+n^{-1/2})^k\}.
    \end{align*}
    By \eqref{A.5}, we have
    \begin{align*}
        r_i^2r_{T,i}^{-2}=\frac{(\X_i-\bmu)^{\T}(\X_i-\bmu)}{(\X_i-\bmu)^{\T}\V_0(\X_i-\bmu)}=\frac{d}{\tr(\bms)}\frac{\z_i^{\T}\bms\z_i}{\z_i^{\T}\z_i},
    \end{align*}
    which, by Cauchy Schwarz inequality, leads to
    \begin{align*}
        \frac{1}{n}\sum_{i=1}^n|r_i^2r_{T,i}^{-2}|^k\le&\left\{\frac{d}{\tr(\bms)}\right\}^k\left(\frac{1}{n}\sum_{i=1}^n\left|\frac{\z_i^{\T}\bms\z_i}{d}\right|^{2k}\right)^{1/2}\left(\frac{1}{n}\sum_{i=1}^n\left|\frac{d}{\z_i^{\T}\z_i}\right|^{2k}\right)^{1/2}=O_p(1),
    \end{align*}
    where the last step follows from \eqref{A.6} and \eqref{A.7}. Further, since
    \begin{align*}
        |1-r_{T,i}^2\hat{r}_{T,i}^{-2}|=\left|1-\frac{r_i^{-2}r_{T,i}^2}{r_i^{-2}\hat{r}_{T,i}^2}\right|\le|r_i^{-2}\hat{r}_{T,i}^2-r_i^{-2}r_{T,i}^2|\cdot\left|r_i^{-2}r_{T,i}^2-\max_{1\le i\le n}|r_i^{-2}\hat{r}_{T,i}^2-r_i^{-2}r_{T,i}^2|\right|^{-1},
    \end{align*}
    we have
    \begin{align*}
        &\frac{1}{n}\sum_{i=1}^n|1-r_{T,i}^2\hat{r}_{T,i}^{-2}|^k\\
        \le&\left(\frac{1}{n}\sum_{i=1}^n|r_i^{-2}\hat{r}_{T,i}^2-r_i^{-2}r_{T,i}^2|^{2k}\right)^{1/2}\left(\frac{1}{n}\sum_{i=1}^n\left|r_i^{-2}r_{T,i}^2-\max_{1\le i\le n}|r_i^{-2}\hat{r}_{T,i}^2-r_i^{-2}r_{T,i}^2|\right|^{-2k}\right)^{1/2}\\
        =&O_p\{(a_n+n^{-1/2})^k\}.
    \end{align*}
    To proceed, since
    \begin{align*}
        r_i^2r_{T,i}^{-2}-r_i^2\hat{r}_{T,i}^{-2}=(r_i^{-2}\hat{r}_{T,i}^2-r_i^{-2}r_{T,i}^2)(r_i^2r_{T,i}^{-2})^2-(r_i^{-2}\hat{r}_{T,i}^2-r_i^{-2}r_{T,i}^2)(r_i^2r_{T,i}^{-2})^2(1-r_{T,i}^2\hat{r}_{T,i}^{-2}),
    \end{align*}
    we have
    \begin{align}
        &\left|\frac{1}{n}\sum_{i=1}^nr_i^2r_{T,i}^{-2}-\frac{1}{n}\sum_{i=1}^nr_i^2\hat{r}_{T,i}^{-2}\right|\nonumber\\
        \le&\left(\frac{1}{n}\sum_{i=1}^n|r_i^{-2}\hat{r}_{T,i}^2-r_i^{-2}r_{T,i}^2|^2\right)^{1/2}\left(\frac{1}{n}\sum_{i=1}^n|r_i^2r_{T,i}^{-2}|^4\right)^{1/2}\nonumber\\
        &+\left(\frac{1}{n}\sum_{i=1}^n|r_i^{-2}\hat{r}_{T,i}^2-r_i^{-2}r_{T,i}^2|^2\right)^{1/2}\left(\frac{1}{n}\sum_{i=1}^n|r_i^2r_{T,i}^{-2}|^8\right)^{1/4}\left(\frac{1}{n}\sum_{i=1}^n|1-r_{T,i}^2\hat{r}_{T,i}^{-2}|^4\right)^{1/4}\nonumber\\
        =&O_p(a_n+n^{-1/2})=o_p(1),\label{A.9}
    \end{align}
    and
    \begin{align*}
        \frac{1}{n}\sum_{i=1}^nr_i^2\hat{r}_{T,i}^{-2}=O_p(1).
    \end{align*}
\end{proof}

\begin{lemma}\label{lemma 8}
    Recalling the definitions of $\z_i$ given in \eqref{A.5}, $r_i$ and $\hat{r}_{T,i},r_{T,i}$ given in \eqref{A.8}, suppose the conditions in Lemma \ref{lemma 7} hold, then
    \begin{align*}
        &\frac{1}{n}\sum_{i=1}^n\left|\frac{\sqrt{d}\xi_i}{\hat{r}_{T,i}^2\Vert\z_i\Vert_2}\right|^2=O_p(\zeta_4^{1/2}),\\
        &\frac{1}{n}\sum_{i=1}^n\left|\frac{d\xi_i^2}{\hat{r}_{T,i}^2\Vert\z_i\Vert_2^2}-\frac{d}{\tr(\bms)}\right|^2=O_p\{(a_n+n^{-1/2}+d^{-1/2})^2\}.
    \end{align*}
\end{lemma}
\begin{proof}
    By the fact that $n^{-1}\sum_{i=1}^nr_i^{-4}=\zeta_4\{1+O_p(n^{-1/2})\}$ and \eqref{A.6}, we have
    \begin{align*}
        &\left|\frac{1}{n}\sum_{i=1}^n\left|\frac{\sqrt{d}\xi_i}{r_{T,i}^2\Vert\z_i\Vert_2}\right|^2-\frac{1}{n}\sum_{i=1}^n\left|\frac{\sqrt{d}\xi_i}{\hat{r}_{T,i}^2\Vert\z_i\Vert_2}\right|^2\right|\\
        =&\left|\frac{1}{n}\sum_{i=1}^nr_i^{-2}\left(\frac{\sqrt{d}\xi_i}{r_i\Vert\z_i\Vert_2}\right)^2r_i^4(r_{T,i}^{-4}-\hat{r}_{T,i}^{-4})\right|\\
        \le&\left(\frac{1}{n}\sum_{i=1}^nr_i^{-4}\right)^{1/2}\left(\frac{1}{n}\sum_{i=1}^n\left|\frac{\sqrt{d}\xi_i}{r_i\Vert\z_i\Vert_2}\right|^4\cdot r_i^8|r_{T,i}^{-4}-\hat{r}_{T,i}^{-4}|^2\right)^{1/2}\\
        \le&\left(\frac{1}{n}\sum_{i=1}^nr_i^{-4}\right)^{1/2}\left(\frac{1}{n}\sum_{i=1}^n\left|\frac{\sqrt{d}\xi_i}{r_i\Vert\z_i\Vert_2}\right|^8\right)^{1/4}\left(\frac{1}{n}\sum_{i=1}^n|r_i^{4}r_{T,i}^{-4}-r_i^{4}\hat{r}_{T,i}^{-4}|^4\right)^{1/4}\\
        =&O_p\left\{\zeta_4^{1/2}\left(\frac{1}{n}\sum_{i=1}^n|r_i^{4}r_{T,i}^{-4}-r_i^{4}\hat{r}_{T,i}^{-4}|^4\right)^{1/4}\right\}.
    \end{align*}
    Using the similar argument as those used in the proof of \eqref{A.9}, we have
    \begin{align*}
        \frac{1}{n}\sum_{i=1}^n|r_i^{4}r_{T,i}^{-4}-r_i^{4}\hat{r}_{T,i}^{-4}|^4=o_p(1).
    \end{align*}
    Hence,
    \begin{align*}
        \left|\frac{1}{n}\sum_{i=1}^n\left|\frac{\sqrt{d}\xi_i}{r_{i,T}^2\Vert\z_i\Vert_2}\right|^2-\frac{1}{n}\sum_{i=1}^n\left|\frac{\sqrt{d}\xi_i}{\hat{r}_{i,T}^2\Vert\z_i\Vert_2}\right|^2\right|=o_p(\zeta_4^{1/2}).
    \end{align*}
    Furthermore, by the fourth conclusion in Lemma \ref{lemma 7} and \eqref{A.6}, we have
    \begin{align*}
        \frac{1}{n}\sum_{i=1}^n\left|\frac{\sqrt{d}\xi_i}{r_{i,T}^2\Vert\z_i\Vert_2}\right|^2=&\frac{1}{n}\sum_{i=1}^nr_i^{-2}|r_i^2r_{i,T}^{-2}|^2\left|\frac{\sqrt{d}\xi_i}{r_i\Vert\z_i\Vert_2}\right|^2\\
        \le&\left(\frac{1}{n}\sum_{i=1}^nr_i^{-4}\right)^{1/2}\left(\frac{1}{n}\sum_{i=1}^n|r_i^2r_{i,T}^{-2}|^8\right)^{1/4}\left(\frac{1}{n}\sum_{i=1}^n\left|\frac{\sqrt{d}\xi_i}{r_i\Vert\z_i\Vert_2}\right|^8\right)^{1/4}\\
        =&O_p(\zeta_4^{1/2}).
    \end{align*}
    Accordingly,
    \begin{align*}
        \frac{1}{n}\sum_{i=1}^n\left|\frac{\sqrt{d}\xi_i}{\hat{r}_{i,T}^2\Vert\z_i\Vert_2}\right|^2=O_p(\zeta_4^{1/2}).
    \end{align*}

    Next, we turn to show the second result. Note that
    \begin{align*}
        \frac{1}{n}\sum_{i=1}^n\left|\frac{d\xi_i^2}{r_{i,T}^2\Vert\z_i\Vert_2^2}-\frac{d}{\tr(\bms)}\right|^2\le&\left\{\frac{1}{n}\sum_{i=1}^n\left|\frac{d\xi_i^2}{r_{i,T}^2\Vert\z_i\Vert_2^2}-\frac{d}{\tr(\bms)}\right|\right\}^2\\
        =&\left\{\frac{1}{n}\sum_{i=1}^n\left|\frac{r_{i,T}^2\Vert\z_i\Vert_2^2}{d\xi_i^2}-\frac{\tr(\bms)}{d}\right|\cdot\left|\frac{d\xi_i^2}{r_{i,T}^2\Vert\z_i\Vert_2^2}\frac{d}{\tr(\bms)}\right|\right\}^2\\
        \le&\left\{\frac{1}{n}\sum_{i=1}^n\left|\frac{r_{i,T}^2\Vert\z_i\Vert_2^2}{d\xi_i^2}-\frac{\tr(\bms)}{d}\right|^2\right\}\cdot\left\{\frac{1}{n}\sum_{i=1}^n\left|\frac{d\xi_i^2}{r_{i,T}^2\Vert\z_i\Vert_2^2}\frac{d}{\tr(\bms)}\right|^2\right\}.
    \end{align*}
    By \eqref{A.5} and the fact that 
    \begin{align}
        \E\left\{\left(\frac{\Vert\z_i\Vert_2^2}{d}-1\right)^{2k}\right\}=\frac{1}{d^{2k}}\sum_{j=0}^{2k}\binom{2k}{j}(-d)^{2k-j}2^j\frac{\Gamma(d/2+j)}{\Gamma(d/2)},\label{A.10}
    \end{align}
    we have
    \begin{align*}
        \frac{1}{n}\sum_{i=1}^n\left|\frac{r_{i,T}^2\Vert\z_i\Vert_2^2}{d\xi_i^2}-\frac{\tr(\bms)}{d}\right|^2=\left\{\frac{\tr(\bms)}{d}\right\}^2\frac{1}{n}\sum_{i=1}^n\left(\frac{\Vert\z_i\Vert_2^2}{d}-1\right)^2=O_p(d^{-1}).
    \end{align*}
    By the fourth conclusion in Lemma \ref{lemma 7} and \eqref{A.6}, we have
    \begin{align*}
        \frac{1}{n}\sum_{i=1}^n\left|\frac{d\xi_i^2}{r_{i,T}^2\Vert\z_i\Vert_2^2}\frac{d}{\tr(\bms)}\right|^2=&\left\{\frac{d}{\tr(\bms)}\right\}^2\frac{1}{n}\sum_{i=1}^n|r_i^2r_{i,T}^{-2}|^2\cdot\left|\frac{\sqrt{d}\xi_i}{r_i\Vert\z_i\Vert_2}\right|^4\\
        \le&\left\{\frac{d}{\tr(\bms)}\right\}^2\left(\frac{1}{n}\sum_{i=1}^n|r_i^2r_{i,T}^{-2}|^4\right)^{1/2}\left(\frac{1}{n}\sum_{i=1}^n\left|\frac{\sqrt{d}\xi_i}{r_i\Vert\z_i\Vert_2}\right|^8\right)^{1/2}\\
        =&O_p(1).
    \end{align*}
    Hence,
    \begin{align*}
        \frac{1}{n}\sum_{i=1}^n\left|\frac{d\xi_i^2}{r_{i,T}^2\Vert\z_i\Vert_2^2}-\frac{d}{\tr(\bms)}\right|^2=O_p(d^{-1}),
    \end{align*}
    which, combining with the third conclusion in Lemma \ref{lemma 7} and \eqref{A.6}, yields 
    \begin{align*}
        &\left|\frac{1}{n}\sum_{i=1}^n\left|\frac{d\xi_i^2}{r_{i,T}^2\Vert\z_i\Vert_2^2}-\frac{d}{\tr(\bms)}\right|^2-\frac{1}{n}\sum_{i=1}^n\left|\frac{d\xi_i^2}{\hat{r}_{i,T}^2\Vert\z_i\Vert_2^2}-\frac{d}{\tr(\bms)}\right|^2\right|\\
        =&\left|\frac{1}{n}\sum_{i=1}^n\left\{\frac{d\xi_i^2}{r_{i,T}^2\Vert\z_i\Vert_2^2}+\frac{d\xi_i^2}{\hat{r}_{i,T}^2\Vert\z_i\Vert_2^2}-\frac{2d}{\tr(\bms)}\right\}\cdot\left(\frac{d\xi_i^2}{r_{i,T}^2\Vert\z_i\Vert_2^2}-\frac{d\xi_i^2}{\hat{r}_{i,T}^2\Vert\z_i\Vert_2^2}\right)\right|\\
        =&\left|\frac{2}{n}\sum_{i=1}^n\left\{\frac{d\xi_i^2}{r_{i,T}^2\Vert\z_i\Vert_2^2}-\frac{d}{\tr(\bms)}\right\}\frac{d\xi_i^2}{r_{i,T}^2\Vert\z_i\Vert_2^2}\left(1-\frac{r_{i,T}^2}{\hat{r}_{i,T}^2}\right)-\frac{1}{n}\sum_{i=1}^n\left|\frac{d\xi_i^2}{r_{i,T}^2\Vert\z_i\Vert_2^2}\right|^2\left(1-\frac{r_{i,T}^2}{\hat{r}_{i,T}^2}\right)^2\right|\\
        \le&2\left\{\frac{1}{n}\sum_{i=1}^n\left|\frac{d\xi_i^2}{r_{i,T}^2\Vert\z_i\Vert_2^2}-\frac{d}{\tr(\bms)}\right|^2\right\}^{1/2}\left(\frac{1}{n}\sum_{i=1}^n\left|\frac{d\xi_i^2}{r_{i,T}^2\Vert\z_i\Vert_2^2}\right|^4\right)^{1/4}\left(\frac{1}{n}\sum_{i=1}^n\left|1-\frac{r_{i,T}^2}{\hat{r}_{i,T}^2}\right|^4\right)^{1/4}\\
        &+\left(\frac{1}{n}\sum_{i=1}^n\left|\frac{d\xi_i^2}{r_{i,T}^2\Vert\z_i\Vert_2^2}\right|^4\right)^{1/2}\left(\frac{1}{n}\sum_{i=1}^n\left|1-\frac{r_{i,T}^2}{\hat{r}_{i,T}^2}\right|^4\right)^{1/2}\\
        =&O_p(d^{-1/2}(a_n+n^{-1/2})+(a_n+n^{-1/2})^2)=O_p\{(a_n+n^{-1/2}+d^{-1/2})^2\}.
    \end{align*}
    Therefore,
    \begin{align*}
        \frac{1}{n}\sum_{i=1}^n\left|\frac{d\xi_i^2}{\hat{r}_{i,T}^2\Vert\z_i\Vert_2^2}-\frac{d}{\tr(\bms)}\right|^2=O_p\{(a_n+n^{-1/2}+d^{-1/2})^2\}.
    \end{align*}
\end{proof}

\subsection{Proof of Theorem 1}
\setcounter{equation}{0}
\renewcommand{\theequation}{B.\arabic{equation}}
\begin{proof}
    \textbf{\emph{Step \uppercase\expandafter{\romannumeral1}. Show that $\Vert\hat{\bms}_0-\bms_0\Vert_{\max}=O_p\left(\sqrt{\frac{\log d}{n}}+\sqrt{\frac{\log n}{n}}\right)$.}} \\
    Recalling the definitions of $\hat{r}_i$, $\z_i$ and $\tilde{\X}_i$ in \eqref{A.1} and \eqref{A.4}, let 
    \begin{align*}
        \hat{\X}_i:=\sqrt{d}U(\X_i-\hat{\bmu})=\frac{\sqrt{d}}{\hat{r}_i}\left(\bmu-\hat{\bmu}+\frac{\xi_i}{\Vert\z_i\Vert_2}\tilde{\X}_i\right),
    \end{align*}
    then $\hat{\bms}_0=n^{-1}\sum_{i=1}^n\hat{\X}_i\hat{\X}_i^{\T}$ and 
    \begin{align*}
        \hat{\X}_i\hat{\X}_i^{\T}=&\frac{d}{\hat{r}_i^2}(\bmu-\hat{\bmu})(\bmu-\hat{\bmu})^{\T}\\
        &+\frac{d\xi_i}{\hat{r}_i^2\Vert\z_i\Vert_2}(\bmu-\hat{\bmu})\tilde{\X}_i^{\T}+\frac{d\xi_i}{\hat{r}_i^2\Vert\z_i\Vert_2}\tilde{\X}_i(\bmu-\hat{\bmu})^{\T}\\
        &+\left(\frac{d\xi_i^2}{\hat{r}_i^2\Vert\z_i\Vert_2^2}-\frac{d}{\tr(\bms)}\right)\tilde{\X}_i\tilde{\X}_i^{\T}+\frac{d}{\tr(\bms)}\tilde{\X}_i\tilde{\X}_i^{\T}.
    \end{align*}
    Thus, for $1\le j,k\le d$,
    \begin{align*}
        \hat{\Sigma}_{0,jk}-\tilde{\Sigma}_{0,jk}=&\frac{1}{n}\sum_{i=1}^n\frac{d}{\hat{r}_i^2}(\mu_j-\hat{\mu}_j)(\mu_k-\hat{\mu}_k)\\
        &+\frac{1}{n}\sum_{i=1}^n\frac{d\xi_i}{\hat{r}_i^2\Vert\z_i\Vert_2}(\mu_j-\hat{\mu}_j)\tilde{X}_{ik}+\frac{1}{n}\sum_{i=1}^n\frac{d\xi_i}{\hat{r}_i^2\Vert\z_i\Vert_2}\tilde{X}_{ij}(\mu_k-\hat{\mu}_k)\\
        &+\frac{1}{n}\sum_{i=1}^n\left(\frac{d\xi_i^2}{\hat{r}_i^2\Vert\z_i\Vert_2^2}-\frac{d}{\tr(\bms)}\right)\tilde{X}_{ij}\tilde{X}_{ik}\\
        =&:\uppercase\expandafter{\romannumeral1}_{jk}+\uppercase\expandafter{\romannumeral2}_{jk}+\uppercase\expandafter{\romannumeral3}_{jk}+\uppercase\expandafter{\romannumeral4}_{jk}.
    \end{align*}
    By Lemma \ref{lemma 2} and \ref{lemma 3}, we have
    \begin{align*}
        \max_{1\le j,k\le d}|\uppercase\expandafter{\romannumeral1}_{jk}|\le d\Vert\hat{\bmu}-\bmu\Vert_{\max}^2\cdot\frac{1}{n}\sum_{i=1}^n\hat{r}_i^{-2}=O_p\left(\zeta_2\zeta_1^{-2}\frac{\log d}{n}\right)=O_p\left(\frac{\log d}{n}\right).
    \end{align*}
    By the first result in Lemma \ref{lemma 4}, we have 
    \begin{align}
        \max_{1\le j\le d}\frac{1}{n}\sum_{i=1}^n\tilde{X}_{ij}^2=\bms_{jj}+O_p\left(\sqrt{\frac{\log d}{n}}\right)=O_p(1),\label{B.1}
    \end{align}
    which, combining with the first result in Lemma \ref{lemma 5} and the Cauchy-Schwarz inequality, yields
    \begin{align*}
        &\max_{1\le j,k\le d}\max(|\uppercase\expandafter{\romannumeral2}_{jk}|,|\uppercase\expandafter{\romannumeral3}_{jk}|)\\
        \le&\sqrt{d}\Vert\hat{\bmu}-\bmu\Vert_{\max}\left(\max_{1\le j\le d}\frac{1}{n}\sum_{i=1}^n\tilde{X}_{ij}^2\right)^{1/2}\left(\frac{1}{n}\sum_{i=1}^n\hat{r}_i^{-4}\right)^{1/4}\left(\frac{1}{n}\sum_{i=1}^n\left(\frac{\sqrt{d}\xi_i}{\hat{r}_i\Vert\z_i\Vert_2}\right)^4\right)^{1/4}\\
        =&O_p\left(\zeta_1^{-1}\zeta_4^{1/4}\sqrt{\frac{\log d}{n}}\right)=O_p\left(\sqrt{\frac{\log d}{n}}\right).
    \end{align*}
    By the Bernstein inequality, we have 
    \begin{align}
        \max_{1\le j\le d}\frac{1}{n}\sum_{i=1}^n\tilde{X}_{ij}^4=O_p(1),\label{B.2}
    \end{align}
    which, combining with the second result in Lemma \ref{lemma 5} and the Cauchy-Schwarz inequality, yields
    \begin{align*}
        \max_{1\le j,k\le d}|\uppercase\expandafter{\romannumeral4}_{jk}|\le&\max_{1\le j,k\le d}\left(\frac{1}{n}\sum_{i=1}^n\tilde{X}_{ij}^2\tilde{X}_{ik}^2\right)^{1/2}\left\{\frac{1}{n}\sum_{i=1}^n\left|\frac{d\xi_i^2}{\hat{r}_i^2\Vert\z_i\Vert_2^2}-\frac{d}{\tr(\bms)}\right|^2\right\}^{1/2}\\
        \le&\max_{1\le j\le d}\left(\frac{1}{n}\sum_{i=1}^n\tilde{X}_{ij}^4\right)^{1/2}\left\{\frac{1}{n}\sum_{i=1}^n\left|\frac{d\xi_i^2}{\hat{r}_i^2\Vert\z_i\Vert_2^2}-\frac{d}{\tr(\bms)}\right|^2\right\}^{1/2}\\
        =&O_p\left(\sqrt{\frac{\log d}{n}}+\sqrt{\frac{\log n}{n}}\right).
    \end{align*}
    In summary, 
    \begin{align*}
        \Vert\hat{\bms}_0-\tilde{\bms}_0\Vert_{\max}=O_p\left(\sqrt{\frac{\log d}{n}}+\sqrt{\frac{\log n}{n}}\right),
    \end{align*}
    which, combining with the first result in Lemma \ref{lemma 4}, leads to
    \begin{align*}
        \Vert\hat{\bms}_0-\bms_0\Vert_{\max}=O_p\left(\sqrt{\frac{\log d}{n}}+\sqrt{\frac{\log n}{n}}\right).
    \end{align*}
    
    \noindent\textbf{\emph{Step \uppercase\expandafter{\romannumeral2}. Show that $\Vert \hat{\L}_m\L_m^{-1}-\I_m \Vert_{\max}=O_p\left(\sqrt{\frac{\log d}{n}}+\sqrt{\frac{\log n}{n}}\right)$.}} \\
    By the Wely's theorem and the fact the $\Vert \M\Vert_2\le d\Vert\M\Vert_{\max}$ for $\M\in\mR^{d\times d}$, we have
    \begin{align*}
        \Vert \hat{\L}_m\L_m^{-1}-\I_m \Vert_{\max}=\max_{1\le j\le m}\left|\frac{\lambda_j(\hat{\bms}_0)}{\lambda_j(\bms_0)}-1\right|\le\frac{\Vert\hat{\bms}_0-\bms_0\Vert_2}{\min_{1\le j\le m}\lambda_j(\bms_0)}=O_p\left(\sqrt{\frac{\log d}{n}}+\sqrt{\frac{\log n}{n}}\right).
    \end{align*}
    
    \noindent\textbf{\emph{Step \uppercase\expandafter{\romannumeral3} Show that $\Vert\hat{\G}_m-\G_m\Vert_{\max}=O_p\left(\sqrt{\frac{\log d}{nd}}+\sqrt{\frac{\log n}{nd}}\right)$.}}\\ 
    By the $\sin\theta$ theorem, for $1\le j\le m$,
    \begin{align*}
        \Vert u_j(\hat{\bms}_0)-u_j(\bms_0)\Vert_2\le&\frac{\sqrt{2}\Vert\hat{\bms}_0-\bms_0\Vert_2}{\min\{|\lambda_{j-1}(\hat{\bms}_0)-\lambda_j(\bms_0)|,|\lambda_j(\bms_0)-\lambda_{j+1}(\hat{\bms}_0)|\}}\\
        =&O_p\left(\sqrt{\frac{\log d}{n}}+\sqrt{\frac{\log n}{n}}\right).
    \end{align*}
    Note that $u_j(\hat{\bms}_0)=\{\lambda_j(\hat{\bms}_0)\}^{-1}\hat{\bms}_0u_j(\hat{\bms}_0)$ and $u_j(\bms_0)=\{\lambda_j(\bms_0)\}^{-1}\bms_0u_j(\bms_0)$, using the fact that $\Vert\M\x\Vert_{\max}\le \sqrt{d}\Vert\M\Vert_{\max}\Vert\x\Vert_2$ for $\M\in\mR^{d\times d}$ and $\x\in\mR^d$, we have for $1\le j\le m$,
    \begin{align*}
        \Vert u_j(\hat{\bms}_0)-u_j(\bms_0)\Vert_{\max}\le&\left|\frac{\lambda_j(\hat{\bms}_0)}{\lambda_j(\bms_0)}-1\right|\frac{1}{\lambda_j(\hat{\bms}_0)}\Vert\hat{\bms}_0u_j(\hat{\bms}_0)\Vert_{\max}\\
        &+\frac{1}{\lambda_j(\bms_0)}\Vert(\hat{\bms}_0-\bms_0)u_j(\hat{\bms}_0)\Vert_{\max}\\
        &+\frac{1}{\lambda_j(\bms_0)}\Vert\bms_0\{u_j(\hat{\bms}_0)-u_j(\bms_0)\}\Vert_{\max}\\
        =&O_p\left(\sqrt{\frac{\log d}{nd}}+\sqrt{\frac{\log n}{nd}}\right).
    \end{align*}
    Hence,
    \begin{align*}
        \Vert\hat{\G}_m-\G_m\Vert_{\max}=\max_{1\le j\le m}\Vert u_j(\hat{\bms}_0)-u_j(\bms_0)\Vert_{\max}=O_p\left(\sqrt{\frac{\log d}{nd}}+\sqrt{\frac{\log n}{nd}}\right).
    \end{align*}
\end{proof}

\subsection{Proof of Theorem 2}
\begin{proof}
    Based on the conclusions of Theorem 1, Theorem 2 is a directly result of Theorem 2.1 in \cite{FLW2018generalPOET}. So we omit the detailed proof here.
\end{proof}

\subsection{Proof of Theorem 3}
\setcounter{equation}{0}
\renewcommand{\theequation}{C.\arabic{equation}}
\begin{proof}
    For CLIME method, based on the conclusions of Theorem 1, the desired results directly follow from Theorem 2.2 in \cite{FLW2018generalPOET}.

    Now we turn to analyze GLASSO method. According to the proof of Theorem 2.2 in \cite{FLW2018generalPOET}, if we have $\Vert\hat{\V}_{0u}-\V_{0u}\Vert_{\max}=O_p(w_n)$, then the remaining desired results follow. To show that $\Vert\hat{\V}_{0u}-\V_{0u}\Vert_{\max}=O_p(w_n)$, from a high-level view, we want to first show that $\Vert\tilde{\V}_{0u}-\V_{0u}\Vert_{\max}=O_p(w_n)$, where 
    \begin{align}
        \tilde{V}_{0u}:=\mathop{\arg\min}\limits_{\V=\V^{\T},\V\succ 0,\V_{\mathcal{S}^c}=\bmzero}\{\tr(\hat{\bms}_{0u}\V)-\log\det(\V)+\tau\Vert\V\Vert_{1,1}\}.\label{C.1}
    \end{align}
    We then show $\tilde{\V}_{0u}=\V_{0u}$. Before introducing the main proof, we first define
    \begin{align}
        \bmD:=\tilde{\V}_{0u}-\V_{0u},~~\W:=\hat{\bms}_{0u}-\bms_{0u}~~\text{and}~~R(\bmD):=(\V_{0u}+\bmD)^{-1}-\V_{0u}^{-1}+\V_{0u}^{-1}\bmD\V_{0u}^{-1}.\label{C.2}
    \end{align}
    \noindent\textbf{\emph{Step \uppercase\expandafter{\romannumeral1}. Show that $\Vert\tilde{\V}_{0u}-\V_{0u}\Vert_{\max}=O_p(w_n)$.}}\\ 
    Consider
    \begin{align*}
        \nabla_{\V}^2\{\tr(\hat{\bms}_{0u}\V)-\log\det(\V)\}=\V^{-1}\otimes\V^{-1}\succ 0,
    \end{align*}
    then 
    \begin{align*}
        \nabla_{\V_{\mathcal{S}}}^2\{\tr(\hat{\bms}_{0u}\V)-\log\det(\V)\}=(\V^{-1}\otimes\V^{-1})_{\mathcal{S}\mathcal{S}}\succ 0.
    \end{align*}
    By Lagrangian duality, we have 
    \begin{align*}
        \tilde{V}_{0u}=\mathop{\arg\min}\limits_{\V=\V^{\T},\V\succ 0,\V_{\mathcal{S}^c}=\bmzero,\Vert\V\Vert_{1,1}\le C(\tau)}\{\tr(\hat{\bms}_{0u}\V)-\log\det(\V)\}.
    \end{align*}
    Hence the strict convexity of the objective function in $\V_{\mathcal{S}}$ implies the uniqueness of $\tilde{\V}_{0u}$. Besides, we know that
    \begin{align}
        \V_{\mathcal{S}}=\tilde{\V}_{0u,\mathcal{S}}\Longleftrightarrow G(\V_{\mathcal{S}}):= \hat{\bms}_{0u,\mathcal{S}}-(\V^{-1})_{\mathcal{S}}+\tau\Z_{\mathcal{S}}=\bmzero~~\text{with}~~\Z_{\mathcal{S}}\in\partial\Vert\V_{\mathcal{S}}\Vert_{1,1}.\label{C.3}
    \end{align}
    For future reference, we note that $\partial\Vert\V\Vert_{1,1}$, the sub-differential of the norm $\Vert\cdot\Vert_{1,1}$ evaluated at some $\V\in\mR^{d\times d}$, consists the set of all symmetric matrices $\M\in\mR^{d\times d}$ such that
    \begin{align}
        M_{ij}=
        \begin{cases}
            \text{sign}(V_{ij}), &\text{if}~~V_{ij}\ne 0,\\
            \in[-1,1], &\text{if}~~V_{ij}=0.\label{C.4}
        \end{cases}
    \end{align}
    Define the map $F:\mR^{|\mathcal{S}|}\to\mR^{|\mathcal{S}|}$ such that 
    \begin{align*}
        F\{\vec(\bmD_{\mathcal{S}}')\}:=-(\bms_{0u,\mathcal{S}\mathcal{S}}^*)^{-1}\vec\{G(\V_{0u,\mathcal{S}}+\bmD_{\mathcal{S}}')\}+\vec(\bmD_{\mathcal{S}}').
    \end{align*}
    By the definition, we know that
    \begin{align*}
        F\{\vec(\bmD_{\mathcal{S}}')\}=\vec(\bmD_{\mathcal{S}}')\Longleftrightarrow \vec\{G(\V_{0u,\mathcal{S}}+\bmD_{\mathcal{S}}')\}=\bmzero\Longleftrightarrow \bmD_{\mathcal{S}}'=\bmD_{\mathcal{S}}.
    \end{align*}
    Define $r:=2\max(\Vert(\bms_{0u,\mathcal{S}\mathcal{S}}^*)^{-1}\Vert_{\infty},\Vert\bms_{0u}\Vert_{\infty})(\Vert\W\Vert_{\max}+\tau)$ and $B(r):=\{\vec(\bmD_{\mathcal{S}}'):\Vert\bmD_{\mathcal{S}}'\Vert_{\max}\le r\}$. We claim that 
    \begin{align}
        F\{B(r)\}\subseteq B(r).\label{C.5}
    \end{align}
    Since $B(r)$ is a nonempty compact convex set, then by Brouwer fixed-point theorem, we can know that $\vec(\bmD_{\mathcal{S}})\in B(r)$, and hence
    \begin{align*}
        \Vert\bmD\Vert_{\max}=\Vert\bmD_{\mathcal{S}}\Vert_{\max}\le r.
    \end{align*}
    By Lemma \ref{lemma 6}, we have $\Vert\W\Vert_{\max}=O_p(w_n)$. Then by setting $\tau\asymp w_n$, we obtain
    \begin{align*}
        \Vert\tilde{\V}_{0u}-\V_{0u}\Vert_{\max}=O_p(w_n).
    \end{align*}

    Now, let us turn to the proof of the claim \eqref{C.5}. For each $\vec(\bmD_{\mathcal{S}}')\in B(r)$, we have
    \begin{align*}
        F\{\vec(\bmD_{\mathcal{S}}')\}=&-(\bms_{0u,\mathcal{S}\mathcal{S}}^*)^{-1}\vec\{G(\V_{0u,\mathcal{S}}+\bmD_{\mathcal{S}}')\}+\vec(\bmD_{\mathcal{S}}')\\
        =&-(\bms_{0u,\mathcal{S}\mathcal{S}}^*)^{-1}\vec\{\W_{\mathcal{S}}+(\V_{0u}^{-1})_{\mathcal{S}}-\{(\V_{0u}+\bmD')^{-1}\}_{\mathcal{S}}+\tau\Z_{\mathcal{S}}-\bms_{0u,\mathcal{S}\mathcal{S}}^*\vec(\bmD_{\mathcal{S}}')\}\\
        =&-(\bms_{0u,\mathcal{S}\mathcal{S}}^*)^{-1}\vec\{\W_{\mathcal{S}}+\tau\Z_{\mathcal{S}}-R(\bmD')_{\mathcal{S}}\},
    \end{align*}
    where the last step follows from the fact that $\vec(\B\A\B)=(\B\otimes\B)\vec(\A)$ for $\A\in\mR^{d\times d}$ and symmetric $\B\in\mR^{d\times d}$. Hence,
    \begin{align*}
        \Vert F\{\vec(\bmD_{\mathcal{S}}')\}\Vert_{\max}\le\Vert (\bms_{0u,\mathcal{S}\mathcal{S}}^*)^{-1}\Vert_{\infty}\{\Vert\W\Vert_{\max}+\tau+\Vert R(\bmD')\Vert_{\max}\}.
    \end{align*}
    Since
    \begin{align*}
        R(\bmD')=&(\V_{0u}+\bmD')^{-1}-\V_{0u}^{-1}+\V_{0u}^{-1}\bmD'\V_{0u}^{-1}\\
        =&(\I_d+\V_{0u}^{-1}\bmD')^{-1}\V_{0u}^{-1}-\V_{0u}^{-1}+\V_{0u}^{-1}\bmD'\V_{0u}^{-1}\\
        =&\sum_{k=2}^\infty(-1)^k(\V_{0u}^{-1}\bmD')^k\V_{0u}^{-1}\\
        =&(\V_{0u}^{-1}\bmD')^2\M\V_{0u}^{-1},
    \end{align*}
    where $\M:=\sum_{k=0}^\infty(-1)^k(\V_{0u}^{-1}\bmD')^k$, then 
    \begin{align*}
        \Vert R(\bmD')\Vert_{\max}=&\max_{1\le i,j,\le d}|\e_i^{\T}(\V_{0u}^{-1}\bmD')^2\M\V_{0u}^{-1}\e_j|\\
        \le&\max_{1\le i\le d}\Vert\e_i^{\T}\V_{0u}^{-1}\bmD'\Vert_{\max}\cdot \max_{1\le j\le d}\Vert\V_{0u}^{-1}\bmD'\M\V_{0u}^{-1}\e_j\Vert_1\\
        \le&\max_{1\le i\le d}\Vert\e_i^{\T}\V_{0u}^{-1}\Vert_1\cdot\Vert\bmD'\Vert_{\max}\cdot\Vert\V_{0u}^{-1}\M^{\T}\bmD'\V_{0u}^{-1}\Vert_{\infty}\\
        \le&\kappa\Vert\V_{0u}^{-1}\Vert_{\infty}^3\cdot\Vert\bmD'\Vert_{\max}^2\cdot\Vert\M^{\T}\Vert_{\infty}\\
        \le&\kappa\Vert\V_{0u}^{-1}\Vert_{\infty}^3\cdot\Vert\bmD'\Vert_{\max}^2\cdot(1-\kappa\Vert\bmD'\Vert_{\max}\cdot\Vert\V_{0u}^{-1}\Vert_{\infty})^{-1}
    \end{align*}
    where the last second step follows from the fact that $\Vert\A\B\Vert_{\infty}\le\Vert\A\Vert_{\infty}\Vert\B\Vert_{\infty}$ for $\A,\B\in\mR^{d\times d}$, and the last step follows from the fact that $\Vert\M^{\T}\Vert_{\infty}\le \sum_{k=0}^\infty(-1)^k\Vert\bmD'\V_{0u}^{-1}\Vert_{\infty}^k\le (1-\Vert\bmD'\Vert_{\infty}\cdot\Vert\V_{0u}^{-1}\Vert_{\infty})^{-1}\le (1-\kappa\Vert\bmD'\Vert_{\max}\cdot\Vert\V_{0u}^{-1}\Vert_{\infty})^{-1}$. Thus, we have
    \begin{align}
        \Vert R(\bmD')\Vert_{\max}\le&\kappa r^2\Vert\V_{0u}^{-1}\Vert_{\infty}^3(1-\kappa r\Vert\V_{0u}^{-1}\Vert_{\infty})^{-1},\label{C.6}
    \end{align}
    which, using the fact that $r\le(3\kappa)^{-1}\min(\Vert(\bms_{0u,\mathcal{S}\mathcal{S}}^*)^{-1}\Vert_{\infty}^{-1}, \Vert\V_{0u}^{-1}\Vert_{\infty}^{-1})$, implies 
    \begin{align*}
        \Vert F\{\vec(\bmD_{\mathcal{S}}')\}\Vert_{\max}\le\frac{r}{2}+\Vert (\bms_{0u,\mathcal{S}\mathcal{S}}^*)^{-1}\Vert_{\infty}\kappa r^2\Vert\V_{0u}^{-1}\Vert_{\infty}^3(1-\kappa r\Vert\V_{0u}^{-1}\Vert_{\infty})^{-1}\le r.
    \end{align*}
    That is $F\{\vec(\bmD_{\mathcal{S}}')\}\in B(r)$. Above all, we show the rate of $\Vert\tilde{\V}_{0u}-\V_{0u}\Vert_{\max}$.

    \noindent\textbf{\emph{Step \uppercase\expandafter{\romannumeral2} Show that $\tilde{\V}_{0u}=\hat{\V}_{0u}$.}}\\
    Similar as the above argument for $\tilde{\V}_{0u}$, we know the uniqueness of $\hat{\V}_{0u}$ and
    \begin{align}
        \V=\hat{\V}_{0u}\Longleftrightarrow  \hat{\bms}_{0u}-\V^{-1}+\tau\Z=\bmzero~~\text{with}~~\Z\in\partial\Vert\V\Vert_{1,1}.\label{C.7}
    \end{align}
    By the definition of $\tilde{\V}_{0u}$ and \eqref{C.3}, we have
    \begin{align*}
        \hat{\bms}_{0u,\mathcal{S}}-(\tilde{\V}_{0u}^{-1})_{\mathcal{S}}+\tau\tilde{\Z}_{\mathcal{S}}=\bmzero,
    \end{align*}
    for some $\tilde{\Z}_{\mathcal{S}}\in\partial\Vert\tilde{\V}_{0u,\mathcal{S}}\Vert_{1,1}$. Setting $\tilde{\Z}_{\mathcal{S}^c}=\tau^{-1}\{-\hat{\bms}_{0u,\mathcal{S}^c}+(\tilde{\V}_{0u}^{-1})_{\mathcal{S}^c}\}$, noting that $\tilde{\V}_{0u,\mathcal{S}^c}=\bmzero$, we have
    \begin{align*}
        \hat{\bms}_{0u}-\tilde{\V}_{0u}^{-1}+\tau\tilde{\Z}=\bmzero,
    \end{align*}
    which, by \eqref{C.7}, yields that to show $\tilde{\V}_{0u}=\hat{\V}_{0u}$, it suffices to show $\tilde{\Z}\in\partial\Vert\tilde{\V}_{0u}\Vert_{1,1}$, i.e., $\Vert\tilde{\Z}_{\mathcal{S}^c}\Vert_{\max}\le 1$.

    Recalling the definitions of $\bmD,\W,R(\bmD)$ in \eqref{C.2}, we have
    \begin{align*}
        \V_{0u}^{-1}\bmD\V_{0u}^{-1}+\W-R(\bmD)+\tau\tilde{\Z}=\bmzero.
    \end{align*}
    Vectorize the equation, we have
    \begin{align*}
        \bms_{0u}^*\vec(\bmD)+\vec(\W)-\vec\{R(\bmD)\}+\tau\vec(\tilde{\Z})=\bmzero,
    \end{align*}
    which, due to $\bmD_{\mathcal{S}^c}=\bmzero$, leads to
    \begin{align*}
        &\bms_{0u,\mathcal{S}\mathcal{S}}^*\vec(\bmD)_{\mathcal{S}}+\vec(\W)_{\mathcal{S}}-\vec\{R(\bmD)\}_{\mathcal{S}}+\tau\vec(\tilde{\Z})_{\mathcal{S}}=\bmzero,~~\text{and}\\
        &\bms_{0u,\mathcal{S}^c\mathcal{S}}^*\vec(\bmD)_{\mathcal{S}^c}+\vec(\W)_{\mathcal{S}^c}-\vec\{R(\bmD)\}_{\mathcal{S}^c}+\tau\vec(\tilde{\Z})_{\mathcal{S}^c}=\bmzero.
    \end{align*}
    Furthermore, we get
    \begin{align*}
        \vec(\bmD)_{\mathcal{S}}=&-(\bms_{0u,\mathcal{S}\mathcal{S}}^*)^{-1}\{\vec(\W)_{\mathcal{S}}-\vec\{R(\bmD)\}_{\mathcal{S}}+\tau\vec(\tilde{\Z})_{\mathcal{S}}\},~~\text{and}\\
        \vec(\tilde{\Z})_{\mathcal{S}^c}=&-\tau^{-1}\{\bms_{0u,\mathcal{S}^c\mathcal{S}}^*\vec(\bmD)_{\mathcal{S}^c}+\vec(\W)_{\mathcal{S}^c}-\vec\{R(\bmD)\}_{\mathcal{S}^c}\}\\
        =&\tau^{-1}\bms_{0u,\mathcal{S}^c\mathcal{S}}^*(\bms_{0u,\mathcal{S}\mathcal{S}}^*)^{-1}\{\vec(\W)_{\mathcal{S}}-\vec\{R(\bmD)\}_{\mathcal{S}}+\tau\vec(\tilde{\Z})_{\mathcal{S}}\}\\
        &-\tau^{-1}\{\vec(\W)_{\mathcal{S}^c}-\vec\{R(\bmD)\}_{\mathcal{S}^c}\}.
    \end{align*}
    Hence, since $r\le(3\kappa)^{-1}\Vert\V_{0u}^{-1}\Vert_{\infty}^{-4}(1+8/\alpha)^{-1}$, by \eqref{C.6}, we have
    \begin{align*}
        \Vert\vec(\tilde{\Z})_{\mathcal{S}^c}\Vert_{\max}\le&\tau^{-1}\Vert\bms_{0u,\mathcal{S}^c\mathcal{S}}^*(\bms_{0u,\mathcal{S}\mathcal{S}}^*)^{-1}\Vert_{\infty}\{\Vert\W\Vert_{\max}+\Vert R(\bmD)\Vert_{\max}+\tau\Vert\tilde{\Z}_{\mathcal{S}}\Vert_{\max}\}\\
        &+\tau^{-1}\{\Vert\W\Vert_{\max}+\Vert R(\bmD)\Vert_{\max}\}\\
        \le&\tau^{-1}(1-\alpha)\alpha\tau/4+1-\alpha\le 1-\alpha/2\le 1. 
    \end{align*}
    Therefore, we prove that $\hat{\V}_{0u}=\tilde{\V}_{0u}$. 

    In summary, we obtain $\Vert\hat{\V}_{0u}-\V_{0u}\Vert_{\max}=O_p(w_n)$. 
    
\end{proof}


\subsection{Proof of Theorem 4}
\begin{proof}
    By Theorem 1 and Wely's Theorem, we have
\begin{align*}
    \lambda_j(\hat{\bms}_0)=
    \begin{cases}
        c_jd,&~~\text{if}~~1\le j\le m,\\
        c_jdw_n,&~~\text{if}~~m+1\le j\le d,
    \end{cases}
\end{align*}
with probability tending to one, where $c_j$'s are some positive constants. It is easy to check for $j<m$ and $j>m$, $\lambda_j(\hat{\bms}_0)/\lambda_{j+1}(\hat{\bms}_0)\asymp 1$,  while for $j=m$, $\lambda_m(\hat{\bms}_0)/\lambda_{m+1}(\hat{\bms}_0)\asymp w_n^{-1}\to\infty$, then $\hat{m}_{ER}$ is consistent. 

For $\hat{m}_{GR}$, apply the inequality $c/(1+c)<\log(1+c)<c$ with $c>0$, then for $j<m$ or $j>m$,
\begin{align*}
    \frac{\ln\{1+\lambda_j(\hat{\bms}_0)/V_{j-1}\}}{\ln\{1+\lambda_{j+1}(\hat{\bms}_0)/V_j\}}<&\frac{\lambda_j(\hat{\bms}_0)}{V_{j-1}}\cdot\frac{1+\lambda_{j+1}(\hat{\bms}_0)/V_j}{\lambda_{j+1}(\hat{\bms}_0)/V_j}\\
    =&\frac{\lambda_j(\hat{\bms}_0)}{\lambda_{j+1}(\hat{\bms}_0)}\cdot\frac{V_j+\lambda_{j+1}(\hat{\bms}_0)}{V_j+\lambda_j(\hat{\bms}_0)}\\
    \le&\frac{\lambda_j(\hat{\bms}_0)}{\lambda_{j+1}(\hat{\bms}_0)}=O_p(1).
\end{align*}
On the other hand, for $j=m$,
\begin{align*}
    \frac{\ln\{1+\lambda_m(\hat{\bms}_0)/V_{m-1}\}}{\ln\{1+\lambda_{m+1}(\hat{\bms}_0)/V_m\}}
    >\frac{\lambda_m(\hat{\bms}_0)}{\lambda_{m+1}(\hat{\bms}_0)}\cdot\frac{V_m}{V_{m-1}+\lambda_m(\hat{\bms}_0)}\asymp \frac{1}{w_n}\frac{\min(n,d)dw_n}{\min(n,d)dw_n+d}\to\infty.
\end{align*}
Thus, $\hat{m}_{GR}$ is consistent.
\end{proof}

\subsection{Proof of Theorem 5}
\begin{proof}
    According to Theorem 2 and 3, we have
    \begin{align*}
        \Vert\hat{\V}_{S}-\V_0\Vert_2=O_p(a_n)~~\text{with}~~a_n=
        \begin{cases}
            w_n^{1-v}m_d,&~~\text{if}~~\hat{\V}_{S}=(\hat{\bms}_{0}^{\tau})^{-1}~~\text{or}~~(\hat{\bms}_{0,s}^{\tau})^{-1},\\
            w_n^{1-v}M_d,&~~\text{if}~~\hat{\V}_{S}=\hat{\V}_0~~\text{or}~~\hat{\V}_{0,s}.
        \end{cases}
    \end{align*}
    Recalling the definitions of $\tilde{\X}_i,\z_i$ given in \eqref{A.4}, $r_i$ and $\hat{r}_{T,i},r_{T,i}$ given in \eqref{A.8}, define
    \begin{align*}
        \hat{\X}_{T,i}:=\frac{\sqrt{d}(\X_i-\hat{\bmu})}{\hat{r}_{T,i}}=\frac{\sqrt{d}}{\hat{r}_{T,i}}\left(\bmu-\hat{\bmu}+\frac{\xi_i}{\Vert\z_i\Vert_2}\tilde{\X}_i\right).
    \end{align*}
    Then $\hat{\bms}_T=n^{-1}\sum_{i=1}^n\hat{\X}_{T,i}\hat{\X}_{T,i}^{\T}$ and 
    \begin{align*}
        \hat{\X}_{T,i}\hat{\X}_{T,i}^{\T}=&\frac{d}{\hat{r}_{T,i}^2}(\bmu-\hat{\bmu})(\bmu-\hat{\bmu})^{\T}\\
        &+\frac{d\xi_i}{\hat{r}_{T,i}^2\Vert\z_i\Vert_2}(\bmu-\hat{\bmu})\tilde{\X}_i^{\T}+\frac{d\xi_i}{\hat{r}_{T,i}^2\Vert\z_i\Vert_2}\tilde{\X}_i(\bmu-\hat{\bmu})^{\T}\\
        &+\left(\frac{d\xi_i^2}{\hat{r}_{T,i}^2\Vert\z_i\Vert_2^2}-\frac{d}{\tr(\bms)}\right)\tilde{\X}_i\tilde{\X}_i^{\T}+\frac{d}{\tr(\bms)}\tilde{\X}_i\tilde{\X}_i^{\T}.
    \end{align*}
    Thus, for $1\le j,k\le d$,
    \begin{align*}
        \hat{\Sigma}_{T,jk}-\tilde{\Sigma}_{0,jk}=&\frac{1}{n}\sum_{i=1}^n\frac{d}{\hat{r}_{T,i}^2}(\mu_j-\hat{\mu}_j)(\mu_k-\hat{\mu}_k)\\
        &+\frac{1}{n}\sum_{i=1}^n\frac{d\xi_i}{\hat{r}_{T,i}^2\Vert\z_i\Vert_2}(\mu_j-\hat{\mu}_j)\tilde{X}_{ik}+\frac{1}{n}\sum_{i=1}^n\frac{d\xi_i}{\hat{r}_{T,i}^2\Vert\z_i\Vert_2}\tilde{X}_{ij}(\mu_k-\hat{\mu}_k)\\
        &+\frac{1}{n}\sum_{i=1}^n\left(\frac{d\xi_i^2}{\hat{r}_{T,i}^2\Vert\z_i\Vert_2^2}-\frac{d}{\tr(\bms)}\right)\tilde{X}_{ij}\tilde{X}_{ik}\\
        =&:\uppercase\expandafter{\romannumeral1}_{T,jk}+\uppercase\expandafter{\romannumeral2}_{T,jk}+\uppercase\expandafter{\romannumeral3}_{T,jk}+\uppercase\expandafter{\romannumeral4}_{T,jk}.
    \end{align*}
    By Lemma \ref{lemma 2}, \eqref{A.2} and the last result in Lemma \ref{lemma 7}, we have
    \begin{align*}
        \max_{1\le j,k\le d}|\uppercase\expandafter{\romannumeral1}_{T,jk}|\le d\Vert\hat{\bmu}-\bmu\Vert_{\max}^2\cdot\zeta_1^2\left(\max_{1\le  i\le n}\nu_i^2\right)\cdot\frac{1}{n}\sum_{i=1}^nr_i^2\hat{r}_{T,i}^{-2}=O_p\left(\frac{\log d\cdot\log n}{n}\right).
    \end{align*}
    By Lemma \ref{lemma 2}, \eqref{B.1} and the first result in Lemma \ref{lemma 8}, we have
    \begin{align*}
        &\max_{1\le j,k\le d}\max(|\uppercase\expandafter{\romannumeral2}_{T,jk}|,|\uppercase\expandafter{\romannumeral3}_{T,jk}|)\\
        \le&\sqrt{d}\Vert\hat{\bmu}-\bmu\Vert_{\max}\cdot\left(\max_{1\le j\le d}\frac{1}{n}\sum_{i=1}^n\tilde{X}_{ij}^2\right)^{1/2}\cdot\left(\frac{1}{n}\sum_{i=1}^n\left|\frac{\sqrt{d}\xi_i}{\hat{r}_{T,i}^2\Vert\z_i\Vert_2}\right|^2\right)^{1/2}\\
        =&O_p\left(\zeta_1^{-1}\zeta_4^{1/4}\sqrt{\frac{\log d}{n}}\right)=O_p\left(\sqrt{\frac{\log d}{n}}\right).
    \end{align*}
    By \eqref{B.2} and the second result in Lemma \ref{lemma 8}, we have
    \begin{align*}
        \max_{1\le j,k\le d}|\uppercase\expandafter{\romannumeral4}_{T,jk}|\le&\max_{1\le j,k\le d}\left(\frac{1}{n}\sum_{i=1}^n\tilde{X}_{ij}^2\tilde{X}_{ik}^2\right)^{1/2}\left\{\frac{1}{n}\sum_{i=1}^n\left|\frac{d\xi_i^2}{\hat{r}_{T,i}^2\Vert\z_i\Vert_2^2}-\frac{d}{\tr(\bms)}\right|^2\right\}^{1/2}\\
        \le&\max_{1\le j\le d}\left(\frac{1}{n}\sum_{i=1}^n\tilde{X}_{ij}^4\right)^{1/2}\left\{\frac{1}{n}\sum_{i=1}^n\left|\frac{d\xi_i^2}{\hat{r}_{T,i}^2\Vert\z_i\Vert_2^2}-\frac{d}{\tr(\bms)}\right|^2\right\}^{1/2}\\
        =&O_p(a_n+n^{-1/2}+d^{-1/2}).
    \end{align*}
    In summary, 
    \begin{align*}
        \Vert\hat{\bms}_T-\tilde{\bms}_0\Vert_{\max}=O_p\left(a_n+\sqrt{\frac{\log d}{n}}\right),
    \end{align*}
    which, combining with the first result in Lemma \ref{lemma 4}, assuming that $a_n=O\left(\sqrt{\frac{\log d}{n}}+\sqrt{\frac{\log n}{n}}\right)$, leads to
    \begin{align*}
        \Vert\hat{\bms}_T-\bms_0\Vert_{\max}=O_p\left(a_n+\sqrt{\frac{\log d}{n}}+\sqrt{\frac{\log n}{n}}\right)=O_p\left(\sqrt{\frac{\log d}{n}}+\sqrt{\frac{\log n}{n}}\right).
    \end{align*}
    Accordingly, using the same techniques as those used in the proof of Theorem 1, 2 and 3, we get the remaining results.
\end{proof}

\subsection{Proof of Proposition 6}
\setcounter{equation}{0}
\renewcommand{\theequation}{D.\arabic{equation}}
\begin{proof}
    According to Theorem 5, we only need to show 
    \begin{align*}
        \left|d^{-1}\widehat{\E(r_{T,i}^2)}-d^{-1}\E(r_{T,i}^2)\right|=O_p(w_n+n^{-\frac{\varepsilon}{2+\varepsilon}}).
    \end{align*}
    Note that the Huber function $H_h(\cdot)$ satisfies 
    \begin{align}
        &H_h(x)-H_h(y)\ge (x-y)\mathbbm{1}(x<h,y>-h),\forall x\ge y,\label{E.1}\\
        &|H_h(a)-H_h(b)|\le|a-b|,\forall a,b.\label{E.2}
    \end{align}
    By \eqref{E.1} and the definition of $d^{-1}\widehat{\E(r_{T,i}^2)}$, we have
    \begin{align*}
        \left|d^{-1}\widehat{\E(r_{T,i}^2)}-d^{-1}\E(r_{T,i}^2)\right|\le \frac{1}{\tilde{p}_h}\left|\frac{1}{n}\sum_{i=1}^nH_h\left(d^{-1}\hat{r}_{T,i}^2-d^{-1}\E(r_{T,i}^2)\right)\right|,
    \end{align*}
    where $\tilde{p}_h=n^{-1}\sum_{i=1}^n\gamma_i$ for $\gamma_i$ defined as follows
    \begin{align*}
        \gamma_i=
        \begin{cases}
            \mathbbm{1}\left(d^{-1}\hat{r}_{T,i}^2-d^{-1}\widehat{\E(r_{T,i}^2)}<h,d^{-1}\hat{r}_{T,i}^2-d^{-1}\E(r_{T,i}^2)>-h\right), &~~\text{if}~~\widehat{\E(r_{T,i}^2)}\le \E(r_{T,i}^2),\\
            \mathbbm{1}\left(d^{-1}\hat{r}_{T,i}^2-d^{-1}\widehat{\E(r_{T,i}^2)}>-h,d^{-1}\hat{r}_{T,i}^2-d^{-1}\E(r_{T,i}^2)<h\right), &~~\text{if}~~\widehat{\E(r_{T,i}^2)}\ge \E(r_{T,i}^2).
        \end{cases}
    \end{align*}   

    \noindent\emph{\textbf{Step \uppercase\expandafter{\romannumeral1}. Analyze $\left|\frac{1}{n}\sum_{i=1}^nH_h\left(d^{-1}\hat{r}_{T,i}^2-d^{-1}\E(r_{T,i}^2)\right)\right|$.}}\\
    Write $d^{-1}\hat{r}_{T,i}^2$ as 
    \begin{align*}
        d^{-1}\hat{r}_{T,i}^2=d^{-1}(\hat{r}_{T,i}^2-r_{T,i}^2)+d^{-1}r_{T,i}^2(1-d^{-1}\Vert\z_i\Vert_2^2)+d^{-2}r_{T,i}^2\Vert\z_i\Vert_2^2=:\uppercase\expandafter{\romannumeral1}_i+\uppercase\expandafter{\romannumeral2}_i+\uppercase\expandafter{\romannumeral3}_i.
    \end{align*}
    Then, by \eqref{E.2}, we have
    \begin{align*}
        \left|\frac{1}{n}\sum_{i=1}^nH_h\left(d^{-1}\hat{r}_{T,i}^2-d^{-1}\E(r_{T,i}^2)\right)\right|\le\left|\frac{1}{n}\sum_{i=1}^nH_h\left(\uppercase\expandafter{\romannumeral3}_i-d^{-1}\E(r_{T,i}^2)\right)\right|+\frac{1}{n}\sum_{i=1}^n|\uppercase\expandafter{\romannumeral1}_i|+\frac{1}{n}\sum_{i=1}^n|\uppercase\expandafter{\romannumeral2}_i|.
    \end{align*}
    By \eqref{A.5} and the assumption that $\E(|\xi/\sqrt{d}|^{2+\varepsilon})<\infty$, we have 
    \begin{align}
        \uppercase\expandafter{\romannumeral3}_i=\frac{\tr(\bms)}{d}\frac{\xi_i^2}{d}\frac{\Vert\z_i\Vert_2^2}{d},~~\E(\uppercase\expandafter{\romannumeral3}_i)=d^{-1}\E(r_{T,i}^2)~~\text{and}~~\E(\uppercase\expandafter{\romannumeral3}_i^{1+\varepsilon/2})<\infty,\label{E.3}
    \end{align}
    which imply that
    \begin{align*}
        \E\left\{H_h\left(\uppercase\expandafter{\romannumeral3}_i-d^{-1}\E(r_{T,i}^2)\right)\right\}=&\E\left\{\left(\uppercase\expandafter{\romannumeral3}_i-d^{-1}\E(r_{T,i}^2)\right)\mathbbm{1}\left(\left|\uppercase\expandafter{\romannumeral3}_i-d^{-1}\E(r_{T,i}^2)\right|\le h\right)\right\}\\
        &+h\E\left\{\sgn\left(\uppercase\expandafter{\romannumeral3}_i-d^{-1}\E(r_{T,i}^2)\right)\mathbbm{1}\left(\left|\uppercase\expandafter{\romannumeral3}_i-d^{-1}\E(r_{T,i}^2)\right|> h\right)\right\}\\
        =&-\E\left\{\left(\uppercase\expandafter{\romannumeral3}_i-d^{-1}\E(r_{T,i}^2)\right)\mathbbm{1}\left(\left|\uppercase\expandafter{\romannumeral3}_i-d^{-1}\E(r_{T,i}^2)\right|> h\right)\right\}\\
        &+h\E\left\{\sgn\left(\uppercase\expandafter{\romannumeral3}_i-d^{-1}\E(r_{T,i}^2)\right)\mathbbm{1}\left(\left|\uppercase\expandafter{\romannumeral3}_i-d^{-1}\E(r_{T,i}^2)\right|> h\right)\right\}\\
        \le&2\E\left\{\left|\uppercase\expandafter{\romannumeral3}_i-d^{-1}\E(r_{T,i}^2)\right|\mathbbm{1}\left(\left|\uppercase\expandafter{\romannumeral3}_i-d^{-1}\E(r_{T,i}^2)\right|> h\right)\right\}\\
        \le&2\E\left(\left|\uppercase\expandafter{\romannumeral3}_i-d^{-1}\E(r_{T,i}^2)\right|^{1+\varepsilon/2}h^{-\varepsilon/2}\right)=O(h^{-\varepsilon/2}).
    \end{align*}
    Furthermore, by Bernstein's inequality and the assumption that $h=cn^{2/(2+\varepsilon)}$ for some constant $c>0$, we obtain
    \begin{align*}
        \left|\frac{1}{n}\sum_{i=1}^nH_h\left(\uppercase\expandafter{\romannumeral3}_i-d^{-1}\E(r_{T,i}^2)\right)\right|=O_p(h^{-\varepsilon/2})=O_p(n^{-\varepsilon/(2+\varepsilon)}).
    \end{align*}
    By the fact that $|\u^{\T}\M\v|\le \Vert\M\Vert_2\Vert\u\Vert_2\Vert\v\Vert_2$ for $\M\in\mR^{d\times d}$ and $\u,\v\in\mR^d$, we have
    \begin{align*}
        |\uppercase\expandafter{\romannumeral1}_i|=&d^{-1}|\hat{r}_{T,i}^2-r_{T,i}^2|\\
        =&d^{-1}|(\X_i-\bmu)^{\T}(\hat{\V}_S-\V_0)(\X_i-\bmu)+(\X_i-\bmu)^{\T}\hat{\V}_S(\bmu-\hat{\bmu})+(\bmu-\hat{\bmu})^{\T}\hat{\V}_S(\X_i-\bmu)|\\
        \le&\Vert\hat{\V}_S-\V_0\Vert_2\cdot d^{-1}r_i^2+2\Vert\hat{\V}_S\Vert_2\cdot\zeta_1\Vert\hat{\bmu}-\bmu\Vert_2\cdot\nu_i\cdot d^{-1}r_i^2.
    \end{align*}
    By H\"{o}lder inequality, \eqref{A.6}, \eqref{A.7} and the assumption that $\E(|\xi/\sqrt{d}|^{2+\varepsilon})<\infty$, we have for $0<\varepsilon'<\varepsilon$,
    \begin{align*}
        &\frac{1}{n}\sum_{i=1}^nd^{-1}r_i^2\\
        =&\frac{1}{n}\sum_{i=1}^n\frac{\xi_i^2}{d}\cdot\frac{\z_i^{\T}\bms\z_i}{d}\cdot\frac{d}{\z_i^{\T}\z_i}\\
        \le&\left\{\frac{1}{n}\sum_{i=1}^n\left(\frac{\xi_i^2}{d}\right)^{\frac{2+\varepsilon'}{2}}\right\}^{\frac{2}{2+\varepsilon'}}\cdot\left\{\frac{1}{n}\sum_{i=1}^n\left(\frac{\z_i^{\T}\bms\z_i}{d}\right)^{\frac{2(2+\varepsilon')}{\varepsilon'}}\right\}^{\frac{\varepsilon'}{2(2+\varepsilon')}}\cdot\left\{\frac{1}{n}\sum_{i=1}^n\left(\frac{d}{\z_i^{\T}\z_i}\right)^{\frac{2(2+\varepsilon')}{\varepsilon'}}\right\}^{\frac{\varepsilon'}{2(2+\varepsilon')}}\\
        =&O_p(1).
    \end{align*}
    Hence,
    \begin{align*}
        \frac{1}{n}\sum_{i=1}^n|\uppercase\expandafter{\romannumeral1}_i|\le&\Vert\hat{\V}_S-\V_0\Vert_2\cdot \frac{1}{n}\sum_{i=1}^nd^{-1}r_i^2+2\Vert\hat{\V}_S\Vert_2\cdot\zeta_1\Vert\hat{\bmu}-\bmu\Vert_2\cdot\max_{1\le i\le n}\nu_i\cdot\frac{1}{n}\sum_{i=1}^n d^{-1}r_i^2\\
        =&O_p\left(a_n+\sqrt{\frac{\log n}{n}}\right).
    \end{align*}
    Similarly, by \eqref{A.10}, we have 
    \begin{align*}
        \frac{1}{n}\sum_{i=1}^n|\uppercase\expandafter{\romannumeral2}_i|=&\frac{\tr(\bms)}{d}\frac{1}{n}\sum_{i=1}^n\frac{\xi_i^2}{d}\cdot\left|1-\frac{\Vert\z_i\Vert_2^2}{d}\right|\\
        \le&\left\{\frac{1}{n}\sum_{i=1}^n\left(\frac{\xi_i^2}{d}\right)^{\frac{2+\varepsilon'}{2}}\right\}^{\frac{2}{2+\varepsilon'}}\cdot\left\{\frac{1}{n}\sum_{i=1}^n\left|1-\frac{\Vert\z_i\Vert_2^2}{d}\right|^{\frac{2+\varepsilon'}{\varepsilon'}}\right\}^{\frac{\varepsilon'}{2+\varepsilon'}}=O_p(d^{-1/2}).
    \end{align*}
    Therefore, we obtain
    \begin{align}
        \left|\frac{1}{n}\sum_{i=1}^nH_h\left(d^{-1}\hat{r}_{T,i}^2-d^{-1}\E(r_{T,i}^2)\right)\right|=O_p\left(n^{-\varepsilon/(2+\varepsilon)}+a_n+\sqrt{\frac{\log n}{n}}+d^{-1/2}\right).\label{E.4}
    \end{align}

    \noindent\emph{\textbf{Step \uppercase\expandafter{\romannumeral1}. Analyze $\tilde{p}_h$.}}\\
    When $\widehat{\E(r_{T,i}^2)}\le \E(r_{T,i}^2)$, we have
    \begin{align*}
        \gamma_i=&\mathbbm{1}\left(d^{-1}\hat{r}_{T,i}^2-d^{-1}\widehat{\E(r_{T,i}^2)}<h,d^{-1}\hat{r}_{T,i}^2-d^{-1}\E(r_{T,i}^2)>-h\right)\\
        \ge&\mathbbm{1}\left(\left|d^{-1}\hat{r}_{T,i}^2-d^{-1}\E(r_{T,i}^2)\right|\le\frac{h}{2},d^{-1}\widehat{\E(r_{T,i}^2)}\ge0\right)\\
        \ge&1-\mathbbm{1}\left(\left|d^{-1}\hat{r}_{T,i}^2-d^{-1}\E(r_{T,i}^2)\right|>\frac{h}{2}\right)-\mathbbm{1}\left(d^{-1}\widehat{\E(r_{T,i}^2)}<0\right),
    \end{align*}
    where the first inequality follows from the fact that $h=cn^{2/(2+\varepsilon)}\to\infty$ and $d^{-1}\E(r_{T,i}^2)<h/2$. Similarly, when $\widehat{\E(r_{T,i}^2)}\ge \E(r_{T,i}^2)$, we have
    \begin{align*}
        \gamma_i\ge&1-\mathbbm{1}\left(\left|d^{-1}\hat{r}_{T,i}^2-d^{-1}\E(r_{T,i}^2)\right|>\frac{h}{2}\right)-\mathbbm{1}\left(d^{-1}\widehat{\E(r_{T,i}^2)}>\frac{h}{2}\right).
    \end{align*}
    Hence, 
    \begin{align}
        \E(\gamma_i)\ge 1-\Pr\left(\left|d^{-1}\hat{r}_{T,i}^2-d^{-1}\E(r_{T,i}^2)\right|>\frac{h}{2}\right)-\Pr\left(d^{-1}\widehat{\E(r_{T,i}^2)}<0\right)-\Pr\left(d^{-1}\widehat{\E(r_{T,i}^2)}>\frac{h}{2}\right).\label{E.5}
    \end{align}
    Note that 
    \begin{align*}
        \Pr\left(\left|d^{-1}\hat{r}_{T,i}^2-d^{-1}\E(r_{T,i}^2)\right|>\frac{h}{2}\right)\le\Pr\left(\left|\uppercase\expandafter{\romannumeral3}_i-d^{-1}\E(r_{T,i}^2)\right|>\frac{h}{6}\right)+\Pr\left(|\uppercase\expandafter{\romannumeral2}_i|>\frac{h}{6}\right)+\Pr\left(|\uppercase\expandafter{\romannumeral1}_i|>\frac{h}{6}\right).
    \end{align*}
    By \eqref{E.3} and Markov's inequality, we have
    \begin{align*}
        \Pr\left(\left|\uppercase\expandafter{\romannumeral3}_i-d^{-1}\E(r_{T,i}^2)\right|>\frac{h}{6}\right)=O\{h^{-(1+\varepsilon/2)}\}=o(1).
    \end{align*}
    Using similar arguments as those used when analyzing $n^{-1}\sum_{i=1}^n|\uppercase\expandafter{\romannumeral1}_i|$ and $n^{-1}\sum_{i=1}^n|\uppercase\expandafter{\romannumeral2}_i|$, by Markov's inequality again, we have
    \begin{align*}
        \Pr\left(|\uppercase\expandafter{\romannumeral2}_i|>\frac{h}{6}\right)+\Pr\left(|\uppercase\expandafter{\romannumeral1}_i|>\frac{h}{6}\right)=o(1).
    \end{align*}
    Hence,
    \begin{align}
        \Pr\left(\left|d^{-1}\hat{r}_{T,i}^2-d^{-1}\E(r_{T,i}^2)\right|>\frac{h}{2}\right)=o(1),\label{E.6}
    \end{align}
    which implies that
    \begin{align*}
        &\Pr\left(d^{-1}\hat{r}_{T,i}^2-d^{-1}\E(r_{T,i}^2)<h\right)\ge 1-\Pr\left(\left|d^{-1}\hat{r}_{T,i}^2-d^{-1}\E(r_{T,i}^2)\right|>\frac{h}{2}\right)\to 1,~~\text{and}\\
        &\Pr\left(d^{-1}\hat{r}_{T,i}^2-d^{-1}\E(r_{T,i}^2)>-h/2\right)\ge 1-\Pr\left(\left|d^{-1}\hat{r}_{T,i}^2-d^{-1}\E(r_{T,i}^2)\right|>\frac{h}{2}\right)\to 1.
    \end{align*}
    Thus,
    \begin{align}
        \Pr\left(\left(h/2-d^{-1}\E(r_{T,i}^2)\right)\frac{1}{n}\sum_{i=1}^n\mathbbm{1}\left(d^{-1}\hat{r}_{T,i}^2-d^{-1}\E(r_{T,i}^2)<h,d^{-1}\hat{r}_{T,i}^2>-h/2\right)>h/4\right)\to 1.\label{E.7}
    \end{align}
    
    By \eqref{E.1} and the fact that $d^{-1}\E(r_{T,i}^2)<h/2$, we have
    \begin{align*}
        &\frac{1}{n}\sum_{i=1}^nH_h\left(d^{-1}\hat{r}_{T,i}^2-h/2\right)\\
        \le&\frac{1}{n}\sum_{i=1}^nH_h\left(d^{-1}\hat{r}_{T,i}^2-d^{-1}\E(r_{T,i}^2)\right)\\
        &-\left(h/2-d^{-1}\E(r_{T,i}^2)\right)\frac{1}{n}\sum_{i=1}^n\mathbbm{1}\left(d^{-1}\hat{r}_{T,i}^2-d^{-1}\E(r_{T,i}^2)<h,d^{-1}\hat{r}_{T,i}^2>-h/2\right),
    \end{align*}
    which, combining with \eqref{E.4} and \eqref{E.7}, yields
    \begin{align*}
        \Pr\left(\frac{1}{n}\sum_{i=1}^nH_h\left(d^{-1}\hat{r}_{T,i}^2-h/2\right)<-h/5\right)\to 1.
    \end{align*}
    Similarly, we can prove
    \begin{align*}
        \Pr\left(\frac{1}{n}\sum_{i=1}^nH_h\left(d^{-1}\hat{r}_{T,i}^2\right)>h/5\right)\to 1.
    \end{align*}
    Furthermore, by the definition of $d^{-1}\widehat{\E(r_{T,i}^2)}$ and the continuity of $H_h(\cdot)$, we obtain
    \begin{align*}
        \Pr\left(0<d^{-1}\widehat{\E(r_{T,i}^2)}<h/2\right)\to 1,
    \end{align*}
    which, combining with \eqref{E.5} and \eqref{E.6}, implies $\E(\gamma_i)=1-o(1)$. Furthermore, by the Markov's inequality and the fact that $\gamma_i\le 1$, we have
    \begin{align*}
        \tilde{p}_h=1+o_p(1).
    \end{align*}

    Combining the conclusions from the two steps, we get the desired result.
\end{proof}

\bibliographystyle{asa}
\bibliography{sn-bibliography}

\end{document}